\documentclass[11pt]{article}
\usepackage{amssymb}%
\usepackage{mathrsfs}
\usepackage[letterpaper]{geometry}
\usepackage{enumerate,color}
\usepackage{hyperref}
\usepackage{amsthm,amsmath,amsfonts}
\usepackage{amsfonts,amsthm,amssymb,amsmath}
\usepackage{graphicx,color,mathrsfs}
\usepackage{hyperref}
\usepackage{tabulary}
\usepackage{float}
\usepackage{enumerate}
\usepackage{tikz}
\usepackage{graphicx}
\usepackage{multicol}
\usepackage{subcaption}
\usepackage{bibspacing}
\newcommand{\R}{\mathbb{R}}

\newcommand{\E}{\mathbb{E}}
\newcommand{\D}{\mathbb{D}}
\newcommand{\Prob}{\mathbb{P}}
\newcommand{\eps}{\epsilon}

\newcommand{\Pb}{\mathbb{P}}

\usepackage[perpage,symbol]{footmisc}
\setfnsymbol{wiley}

\newtheorem{lemma}{Lemma}[section]
\newtheorem{theorem}{Theorem}[section]
\newtheorem{proposition}{Proposition}[section]

\numberwithin{equation}{section}

\newtheorem{assumption}{Assumption}[section]
\newcommand{\red}[1]{\textcolor{black}{#1}}

\newcommand{\add}[1]{\textcolor{black}{#1}}

\begin{document}

\begin{center}
  \Large Hydrodynamic limit of order book dynamics
\end{center}

\author{}
\begin{center}
{Xuefeng
  Gao}\,\footnote{Department of Systems
    Engineering and Engineering Management, The Chinese University of Hong Kong, Shatin, N.T. Hong Kong;
    xfgao@se.cuhk.edu.hk},
   S. J. Deng\,\footnote{H. Milton Stewart School of
    Industrial and Systems Engineering, Georgia Institute of
    Technology, Atlanta, GA 30332; deng@isye.gatech.edu}
\end{center}

\begin{center}
First Version: November 27, 2014\\
This Version: \today
\end{center}

\begin{abstract}
 \red{In this paper, we establish a fluid limit for a two--sided Markov order book model. 
 Our main result states that in a certain asymptotic regime, a pair of measure-valued processes representing the ``sell-side shape'' and ``buy-side shape'' of an order book converges to a pair of deterministic measure-valued processes in a certain sense. 
We also test our fluid approximation on data. The empirical results suggest that the approximation is reasonably good for liquidly--traded stocks in certain time periods.}
\end{abstract}

\section{Introduction}

%

As a trading mechanism, limit order books have gained growing popularity in equity and
derivative markets in the past two decades. Nowadays, the majority of the world's
financial markets, such as Electronic Communication Networks in the United States, the Hong Kong Stock Exchange and the Toronto Stock Exchange, are organized as electronic limit order books to match buyers and sellers. This has inspired intense research activities on limit order books. See, e.g., \cite{gould2013limit, parlour2008limit} for reviews.



In this paper, our primary goal is to develop approximations for the evolution of the shape of a limit
order book over a time horizon that is large compared with the
length of time between order book events. To this end, we
perform asymptotic analysis of a stochastic limit order book
model and thus connect the dynamics of a limit order book on two different time scales. There has been a growing interest in studying such scaling limits of order book dynamics, mainly motivated by a desire to better understand the price formation process and the relation between the parameters of
the point processes describing order flow at high
frequency and the parameters of models describing price and order--book shape dynamics at a larger time scale. See, e.g., \cite{abergel2013stability, blanchet2013continuous, cont2011order, cont2013price,Horsty13orderbook,kirilenko2013multiscale, Josh14orderbook,lakner2013high}.

{We contribute to this body of literature by establishing a fluid limit for a two--sided Markov order book model in the following asymptotic regime: (a) Tick size goes to zero; (b)
Rates of order arrivals goes to infinity; (c) Relative order size (compared to limit order queue size) goes to zero.
Such a regime is relevant for high--frequency trading of
liquid stocks where the tick size is small (one penny in U.S.), inter--arrival times of orders are very short (on the order of milliseconds),
and there are large volumes of limit orders sitting near the market price and waiting
for transactions. }

Besides establishing a scaling limit, we also test our fluid approximation on historical order book data from the U.S. exchange NYSE Arca.
Our in--sample analysis demonstrates that our theory can be potentially useful to
approximate the evolution of a limit order book for liquid stocks in time periods of low price volatility. 

Our asymptotic analysis is built on the order book model in \cite{Cont10}, where the high frequency dynamics of order book events are described by a continuous time Markov chain (CTMC). 
{In their discrete order book model}, there are a finite number of security price levels and the state of the order book is described by an order quantity at each price.
We develop a fluid approximation for this high--dimensional CTMC by establishing a type of law of large numbers limit theorem.

Our main result (Theorem~\ref{lem:ode}) states that as the tick size $1/n$ approaches zero\footnote{The tick size, the arrival rates of orders, and the relative order sizes are all scaled as a power of a
large parameter $n$. See Section 2 for more details.}, (a) the sequence of best bid and best ask prices converges in probability to a constant; and
  (b) the sequence of pairs of measure-valued processes $\{(\zeta^{n, +},  \zeta^{n, -}): n \ge 1\}$, representing the ``sell-side shape'' and ``buy-side shape'' of the order book, converges to a pair of deterministic measure-valued processes in a certain sense. Moreover, the density profile of the limiting processes
can be described by Ordinary Differential Equations (ODEs) with coefficients determined by order flow intensities.

 Our key innovation in establishing the fluid limit is to represent a {\it{two--sided limit order book}} by {\it{a pair of non-negative measure--valued processes}} and studying the weak convergence of such processes using martingale methods.
Such a measure--valued process representation of an order-book and the related proof techniques could be potentially useful for establishing scaling limits for more general two--sided order book models such as multi--scale models allowing the co-existence of high--frequency and low--frequency traders. See, e.g., \cite{cont2015working}.

We now compare our work with some closely related alternatives which also study the scaling limits
 of the shape of limit order books {in} certain asymptotic regimes. These include \cite{bayer2014functional, horst2015weak, Horsty13orderbook, kruk2003functional, kruk2012limiting, Josh14orderbook, lakner2013high}.
Among them, only the three studies \cite{bayer2014functional, horst2015weak, Horsty13orderbook} incorporate order cancellations\footnote{Modelling order cancellations is important: for example, more than 95\% of limit orders at NASDAQ are cancelled without execution, see \cite{hautsch2011limit}.}.
In \cite{Horsty13orderbook}, the authors prove a law of large numbers result for the limit order book dynamics. In their limit, the best bid and best ask price dynamics can be described by two coupled ODEs and the relative buy and sell shape functions can be described by two linear first-order partial differential equations.
This fluid limit result is later generalized in \cite{horst2015weak} to allow order flows depend on the state of the order book including the prices and volumes, and in
\cite{bayer2014functional} which establishes the functional central limit theorems for order book dynamics. We remark that though we rely on a slightly different discrete stochastic order-book model to conduct asymptotic analysis, our theoretical result is not completely new in view of these three studies. However, we emphasize that our proof technique to establish the scaling limit is novel and it is very different from theirs.
Moreover, unlike the aforementioned studies, we empirically test our fluid approximation on real data and illustrate the potential relevance of the approximation. These are the two main contributions of this paper.

Finally, we remark that we have had to focus on a few features to keep the model and the asymptotic analysis tractable, and that this has necessarily led to some undesirable consequences. We briefly discuss the main limitations. First, the limit is deterministic, and the limit price is a constant. 
This is not realistic. However, such a deterministic model could be potentially useful. For example, a similar fluid order book model, where the best bid and ask prices are also constants, has been used in \cite{maglaras2015optimal} to analyze the problem of limit and market order placement to optimally buy
a block of shares over a fixed time horizon in the order of several minutes.
Second, to keep the mathematical analysis tractable, a number of realistic features such as non--stationarity, time--clustering and mutual dependence of order flows are necessarily left out in our asymptotic analysis. For recent developments along this line, see, e.g., \cite{abergel2013stability}.\\

\textbf{Outline of this paper.} Section~\ref{sec:model}
reviews a variant of the Markov order book model in \cite{Cont10} and states the assumptions on order flow rates and initial conditions. Section~\ref{sec:result} summarizes our main result. Section~\ref{sec:booktest} discusses empirical analysis. Section~\ref{sec:prelim} presents preliminaries for proving the main result. Section~\ref{sec:proofthm1} is devoted to the proof of the main result. Auxiliary results and their proofs are provided in the appendices.\\

\textbf{Notation.}
All random elements are defined on a
common probability space $(\Omega,\mathcal{F},\mathbb{P})$ unless
otherwise specified. Given $x\in\mathbb{R}$, we set $x^{+}=\max\{x,0\}$ and
$x^{-}=\max\{-x,0\}$. We write $\lceil x \rceil$ for the smallest integer not less than $x$, and $\lfloor x \rfloor$ for the largest integer not greater than $x$. 
The set of continuous functions on $[0,1]$ is denoted by $C([0,1])$. Given a Polish space $\mathcal{E}$,
the space of right-continuous functions $f: [0,T]
\rightarrow \mathcal{E}$  with left limits is denoted by $\D([0,T], \mathcal{E})$. The space $\D([0,T], \mathcal{E})$ is assumed to be
endowed with the Skorohod $J_1$-topology.
For a sequence of random elements $\{X_n: n \ge 1\}$
taking values in a metric space, we write $X_n  \Rightarrow X$ to
denote the convergence of $X_n$ to $X$ in distribution. Each
stochastic process with sample paths in $\D([0,T], \mathcal{E})$ is considered to
be a $\D([0,T], \mathcal{E})$-valued random element. For a Borel measure $\nu$ and function $f$,
we set $\langle \nu, f \rangle= \int f(u) \nu(du)$ when the integration exists.
The symbol $\delta_u$ represents the Dirac measure at location $u \in \mathbb{R} ,$ i.e., for a Borel set $U$,
\begin{equation*}
\delta_u (U)=
\begin{cases} 1 & \text{if $u\in U$,}
\\
0 &\text{if $u \notin U$}.
\end{cases}
\end{equation*}
The space of finite {\it{non-negative}} measures on $[0,1]$ is denoted by $\mathcal{M^+}([0,1])$, and the space of finite {\it{signed}} measures on $[0,1]$ is denoted by $\mathcal{M}([0,1])$.

\section{Model and Assumptions} \label{sec:model}
In this section we describe a variant of the order book model introduced in \cite{Cont10} and state the assumptions. Throughout of this paper, we fix the time $T>0$.

Fix $n\ge 1$. Without loss of generality, we suppose that investors can submit their limit orders to $n$ discrete equally-spaced price levels $\{\frac{1}{n}, \frac{2}{n}, \cdots, \frac{n}{n}\}$ within the price range $(0,1]$. Thus the parameter $n$ also represents the inverse of tick size. The state of the limit order book at time $t$ is given by a vector $\mathcal{X}^n(t) \equiv (\mathcal{X}^n_1(t), \ldots, \mathcal{X}^n_n(t)) \in \mathbb{Z}^n$,  where for $i \in \{1, \ldots, n\}$, $|\mathcal{X}^n_i(t)|$ represents the number of outstanding limit orders at price $i/n$ at time $t$. If $\mathcal{X}^n_i(t) > 0$, then there are $\mathcal{X}^n_i(t)$ sell orders at price $i/n$, and if $\mathcal{X}^n_i(t) < 0$, then there are $-\mathcal{X}^n_i(t)$ buy orders at price $i/n$.

To define the best ask price (lowest price among limit sell orders) and best bid price (highest price among limit buy orders), we define two mappings $\mathcal{P}^n_A$ and $\mathcal{P}^n_B$ such that
for a given state $x=(x_1, \ldots, x_n) \in \mathbb{Z}^n$
\begin{eqnarray}
\mathcal{P}^n_A(x)&\equiv& \inf\{i \in \{1, \ldots, n\}: x_i>0\}\wedge \{n+1\}, \label{eq:bestaskmap}\\
\mathcal{P}^n_B(x)& \equiv&\sup\{i \in \{1, \ldots, n\}: x_i<0\} \vee  0, \label{eq:bestbidmap}
\end{eqnarray}
where $\inf \emptyset = \infty$ and $\sup \emptyset = -\infty$ by convention. For $t\ge0$, we write
\begin{eqnarray}
p^n_A(t)&=& \mathcal{P}^n_A(\mathcal{X}^n(t)) , \label{eq:bestask}\\
p^n_B(t)&= & \mathcal{P}^n_A(\mathcal{X}^n(t)) . \label{eq:bestbid}
\end{eqnarray}
Thus for order book $\mathcal{X}^n(\cdot)$ at time $t$, the best ask price is $p_A^n(t)/n$ and best bid price is $p^n_B(t)/n$, where $1/n$ is the tick size. We will be interested in the regime where $n \rightarrow \infty$, i.e., the tick size approaches zero. For each fixed $n$, as in \cite{Cont10}, we assume
\begin{itemize}
 \setlength\itemsep{0.05em}
\item [(a)] Limit buy (respectively sell) orders arrive at a distance
of $i$ ticks from the opposite best quote at independent,
exponentially distributed times with rate $\Lambda_B^n(i)$ (respectively $\Lambda_A^n(i)$),
\item [(b)] Market buy (respectively sell) orders arrive at independent,
exponentially distributed times with rate $\Upsilon_B^n$ (respectively $\Upsilon_A^n$),
\item [(c)] Each limit buy (respectively sell) order at a distance of $i$ ticks
from the opposite best quote is cancelled independently after exponentially distributed times with rate $\Theta_B^n(i)$ (respectively $\Theta_A^n(i)$ ) .
\item [(d)] All the above events are mutually independent.
\item [(e)] All the orders are of unit size, which is independent of $n$.
\end{itemize}

\red{Compared to the model in~\cite{Cont10}, the extra feature is that we allow order flow intensities to be side-dependent (buy orders or sell orders).} This has been observed empirically and such feature will be useful in our empirical analysis in Section~\ref{sec:booktest}. Given these assumptions, the state process $\mathcal{X}^n(\cdot)$ is a $n-$dimensional CTMC with state space $\mathbb{Z}^n$. \red{Define an operator $\mathcal{L}_n$ as follows: for any function $\mathcal{H}:  \mathbb{Z}^n\to \R$ and each $x \in \mathbb{Z}^n $}
\begin{eqnarray}\label{eq:generator}
{\mathcal{L}_n \mathcal{H} (x)}
&=&\sum_{k<\mathcal{P}^n_A(x)} \left[\left(\mathcal{H}(x^{k+})-\mathcal{H}(x)\right)\cdot \Theta_B^n(\mathcal{P}^n_A(x)-k)\cdot |x_k| \right] \nonumber \\
&&+ \sum_{k<\mathcal{P}^n_A(x) } \left[ \left(\mathcal{H}(x^{k-})-\mathcal{H}(x)\right) \cdot \Lambda_B^n(\mathcal{P}^n_A(x)-k)\right] \nonumber\\
&&+\sum_{k>\mathcal{P}^n_B(x)} \left[\left(\mathcal{H}(x^{k+})-\mathcal{H}(x)\right) \cdot\Lambda_A^n(k-\mathcal{P}^n_B(x)) \right] \nonumber \\
&&+\sum_{k>\mathcal{P}^n_B(x)} \left[ \left(\mathcal{H}(x^{k-})-\mathcal{H}(x)\right) \cdot \Theta_A^n(k-\mathcal{P}^n_B(x)) \cdot |x_k|\right] \nonumber\\
&&+ \left(\mathcal{H}(x^{\mathcal{P}^n_B(x)^+})-\mathcal{H}(x) \right) \cdot \Upsilon_A^n
\nonumber \\
&&+ \left(\mathcal{H}(x^{\mathcal{P}^n_A(x)-})-\mathcal{H}(x) \right) \cdot \Upsilon_B^n,
\end{eqnarray}
where $\mathcal{P}^n_A$ and $\mathcal{P}^n_B$ are mappings given in \eqref{eq:bestaskmap} and \eqref{eq:bestbidmap}, and for $k \in \{1, \ldots, n\}$
\begin{equation*}
\begin{cases}
x^{k+}=(x_1, \ldots, x_{k-1}, x_k+1, x_{k+1}, \ldots, x_{n} ),\\
x^{k-}=(x_1, \ldots, x_{k-1}, x_k - 1, x_{k+1}, \ldots, x_{n} ), \\
\end{cases}
\end{equation*}
and
\[x^{0+} = x^{(n+1)-} = x.\]
\red{The infinitesimal generator of $\mathcal{X}^n$ coincides with the operator $\mathcal{L}_n$ on the space of bounded functions from $\mathbb{Z}^n$ to $\R$.}

A few remarks are in order. First, the price range $(0, 1]$ is not essential in our asymptotic analysis. We can also use $(c_1, c_2]$ for some integers $c_1< c_2$ that are independent of $n$, and this would not affect our results. Second, the state space of the above model is the $n-$dimensional integer lattice and it becomes infinitely dimensional as $n \rightarrow \infty$. So standard limit theorems on Markov chains (see, e.g. \cite{kurtz1971limit}) do not apply in our asymptotic analysis.
Third, it is possible to allow random order size with finite second moment in our analysis. To keep the presentation simple, we work with unit order size in this paper.

Throughout the paper, we make the following assumptions.
\begin{assumption}
\label{ass:initialcondition} There exists a function
$\varrho$ from $[0,1]$ to $\R$ such that for fixed $n \ge 1,$ the initial order book $\mathcal{X}^n(0)$ is given by
\begin{eqnarray} \label{eq:X0}
\mathcal{X}^n_i(0) \equiv n \cdot \varrho (\frac{i}{n})  \quad \text{for $i \in \{1, 2, \ldots, n\}$}.
\end{eqnarray}
The initial profile $\varrho$ satisfies
\begin{enumerate}
\item [(a)]  For some $\mathsf{p} \in (0,1)$, $\varrho(\mathsf{p})=0$ and
\begin{eqnarray} \label{eq:gammap}
\sup \{x \in [0,1]: \varrho(x) <0 \} = \inf\{x \in [0,1]: \varrho(x)
>0\} =\mathsf{p};
\end{eqnarray}
 \item [(b)] the function $|\varrho|$ is positive on $(\mathsf{p}-\Delta, \mathsf{p})$ and $(\mathsf{p}, \mathsf{p}+\Delta)$ for some small $\Delta>0$;
\item [(c)] the function $\varrho$ is continuous and bounded on $[0, \mathsf{p})$ and $(\mathsf{p}, 1]$.
 \end{enumerate}
\end{assumption}

\begin{assumption}
\label{ass:rateconverg1} For {each market side} $j \in \{A, B\}$, there exist positive continuous function $\Lambda_j(x)$ and nonnegative continuous function $\Theta_j(x)$ on $[0,1]$, and constants $\kappa<1$ and $\Upsilon \ge 0$
such that for each $n$ and $i \in \{1, 2, \ldots, n\}$,
\begin{eqnarray} \label{eq:rates}
{\Lambda_j^n(i)}= \Lambda_j (\frac{i}{n}), \quad
{\Theta_j^n( i )} = \frac{1}{n} \Theta_j (\frac{i}{n}), \quad \text{and}
\quad {\Upsilon_j^n} \le n^{\kappa} \cdot \Upsilon.
\end{eqnarray}
\end{assumption}

Both assumptions are motivated by our data analysis. First, empirically we find that for certain liquidly traded stocks such as Bank of America, Ford Motor Co. and Dell Inc., there are large volumes of limit orders sitting at or near the best quotes and waiting for executions. This provides empirical support for \eqref{eq:X0}. For technical convenience, we also make three assumptions on the regularity of the initial profile $\varrho$. Here, the number $\mathsf{p}$ can be viewed as the `market price' of a stock at time zero.
Second, from our data we also find that for the aforementioned liquid stocks, typically
the ratio of limit order arrival rate and cancel rate
$\frac{\Lambda_j^n(i)} {\Theta_j^n(i)}$ is of the order of $n=100$ in the
vicinity of the best quotes. In addition, we observe that the magnitude of market order arrival rate (e.g., volume per
second) is not very high compared to the volumes of limit orders in the vicinity of market
price, specifically during the time periods of low price volatility. These two facts
motivate \eqref{eq:rates}.
The assumptions on the continuity of functions $\Lambda_j, \Theta_j$ and positivity of $\Lambda_j$ are technical and they facilitate our mathematical analysis.

\section{The main result}\label{sec:result}
In this section we state our main result of this paper. We first define three sequences of measure-valued process (see, e.g., \cite{dawson1993measure} for background on measure-valued processes). For fixed $n \ge 1$ and  $t \in [0, T]$, we set
\begin{eqnarray}
\zeta_t^{n,+} &=& \frac{1}{n^{2}}  \sum_{i=1}^{n} (\mathcal{X}^n_{i} (n t))^+ \cdot
\delta_{\frac{i}{n}} = \frac{1}{n^{2}}  \sum_{i> p_B^n(nt)} {\mathcal{X}^n_{i}} ( nt) \cdot
\delta_{\frac{i}{n}}, \label{eq:empirialshapeplus} \\
\zeta_t^{n,-} &=&  \frac{1}{n^{2}}  \sum_{i=1}^{n} (\mathcal{X}^n_{i} (n t))^- \cdot
\delta_{\frac{i}{n}} = \frac{1}{n^{2}}  \sum_{i< p_A^n(nt)} {-\mathcal{X}^n_{i}} (n t) \cdot
\delta_{\frac{i}{n}},  \label{eq:empirialshapeminus}\\
\zeta_t^n &=& \zeta_t^{n,+} - \zeta_t^{n,-} = \frac{1}{n^{2}}  \sum_{i=1}^{n} \mathcal{X}^n_{i} (n t) \cdot
\delta_{\frac{i}{n}},\label{eq:empirialshape}
\end{eqnarray}
\red{where $\mathcal{X}^n$ is the $n-$dimensional CTMC with state space $\mathbb{Z}^n$ and it describes the evolution of the limit order book}, $p_A^n, p_B^n$ are given in \eqref{eq:bestask} and \eqref{eq:bestbid}, and $\delta_u$ is the Dirac measure centered at $u$. The choice of space scaling $n^{2}$ follows from the facts that there are $n$ price levels in total and the number of limit orders on each price level is of the order $n$ (Assumptions~\ref{ass:initialcondition} and \ref{ass:rateconverg1}).

%
%

For fixed $t$, the pair $(\zeta^{n,+}_{t}, \zeta^{n,-}_{t})$ represents the ``sell-side shape" and ``buy-side shape" of the order book at some large time $nt$, and the measure $\zeta^n_{t}$ represents the whole shape of the order book at large time $nt$.
Note that the pair $(\zeta^{n,+}_{t}, \zeta^{n,-}_{t})$ lives in the product space $\mathcal{M^+}([0,1]) \times \mathcal{M^+}([0,1])$. For the moment, we use $\overline{\mathcal{M}}([0,1])$ to denote this product space equipped with an appropriate metric such that $\overline{\mathcal{M}}([0,1])$ is complete and separable. A precise definition of the space $\overline{\mathcal{M}}([0,1])$ is given in Section~\ref{sec:tight}.

We are interested in the limiting behavior of the Markov process $(\zeta^{n,+}, \zeta^{n,-})$
as $n \rightarrow \infty$. \add{This corresponds to the limiting regime where price tick size $1/n$ approaches zero,
rates of order arrivals goes to infinity, and relative order size (compared to limit order queue size) goes to zero. This can be readily seen from Assumptions~\ref{ass:initialcondition}--\ref{ass:rateconverg1} and the fact that the time is speeded up by $n$ in \eqref{eq:empirialshapeplus}--\eqref{eq:empirialshapeminus}.}
It is worthwhile to note that
$\zeta^n$ is also a Markov process which contains the same information as $(\zeta^{n,+}, \zeta^{n,-})$. However, for fixed $t$, $\zeta^n_{t}$ is a signed measure living in the space $\mathcal{M}([0,1])$ which is ill-suited for studying weak convergence since the weak topology is in general not metrizable (see, e.g., \cite{spiliopoulos2014fluctuation}). Thus we instead work with the sequence of pairs of non-negative measure-valued processes $\{(\zeta^{n,+}, \zeta^{n,-}): n \ge 1 \}$.


\add{Now we state the main theorem of this paper. The proof is given in~Section~\ref{sec:proofthm1}.  Recall that $\varrho, \mathsf{p}, \Lambda_A, \Lambda_B, \Theta_A$ and $\Theta_B$ are given in Assumptions~\ref{ass:initialcondition} and~\ref{ass:rateconverg1}.
\begin{theorem} \label{lem:ode}
Suppose that Assumptions~\ref{ass:initialcondition} and \ref{ass:rateconverg1}
hold. Then {as $n \rightarrow
\infty$,}
\begin{itemize}
\item [(a)] we have
\begin{eqnarray}
 \sup_{0 \le t \le T} \left| \frac{p_A^n(nt)}{n}- \mathsf{p} \right| \Rightarrow 0, \label{eq:askpriceconvg}\\
 \sup_{0 \le t \le T} \left| \frac{p_B^n(nt)}{n}-  \mathsf{p} \right| \Rightarrow 0. \label{eq:bidpriceconvg}
\end{eqnarray}
\item [(b)] we have
 \begin{eqnarray*}\label{eq:resultconvg}
 (\zeta^{n,+}, \zeta^{n,-}) \Rightarrow (\zeta^+, \zeta^-) \quad \text{in $\mathbb{D}([0, T], \overline{\mathcal{M}}([0,1]) )$},
 \end{eqnarray*}
 where $(\zeta^+, \zeta^-)$ is a pair of deterministic measure-valued process.
 In addition, for any $t \in [0, T]$, the nonnegative measures $\zeta^+_t$ and $\zeta^-_t$ are absolutely continuous with respect to Lebesgue measure and have density functions $\varphi^+ = \max\{\varphi, 0\}$ and $\varphi^- = \max\{-\varphi, 0\}$ such that
 $\zeta^{+}_t(dx) = \varphi^{+}(x,t) dx$ and $\zeta^{-}_t(dx) = \varphi^{-}(x,t) dx$ for $x \in [0,1]$. The function $\varphi$ is uniquely determined by the following set of equations: for $x \in [0,1]$
 \begin{eqnarray}
\varphi(x,0) &=& \varrho (x), \label{eq:varphiinitial}\\
\partial_t \varphi(x,t) &=& \Lambda_A ( x - \mathsf{p})  - \Theta_A(x - \mathsf{p}) \cdot \varphi (x,t), \quad \text{$x>\mathsf{p}$, } \label{eq:varphiODE}\\
\partial_t \varphi(x,t) &=& -\Lambda_B ( \mathsf{p}-x) - \Theta_B( \mathsf{p}-x) \cdot \varphi (x,t), \quad \text{$x<\mathsf{p}$ }. \label{eq:varphiODE2}\\
\varphi(\mathsf{p}, t)&=&0, \quad \text{$t \in [0,T].$} \label{eq:varphip}
\end{eqnarray}
\end{itemize}
\end{theorem}
}

Theorem~\ref{lem:ode} provides a fluid approximation for the order book dynamics. Note that in high--frequency trading, the inter--arrival times of order book events are typically very short (in the order of milliseconds). So the result suggests that on a `low-frequency' time scale (in the order of a few minutes),
the scaled best bid and best ask prices are a constant $\mathsf{p}$, and the density profile of the sell-side and buy-side shape of the order book, given by $\varphi^+(x,t)$ and $\varphi^-(x,t)$, are characterized by a linear ODE with coefficients determined by order arrival rates and cancellation rates.
As discussed in the introduction, this deterministic limit may not be realistic, but it could be potentially useful for studying other problems.

 One readily verifies from Theorem~\ref{lem:ode} and \eqref{eq:gammap} that for each $t\in[0,T]$,
\begin{eqnarray}
\varphi^{+}(x,t) = e^{-\Theta_A(x-\mathsf{p}) t}\cdot {\varrho(x)}^{+} +
\frac{\Lambda_A(x-\mathsf{p})}{\Theta_A(x-\mathsf{p})} \cdot \Big(1- e^{-\Theta_A(x-\mathsf{p})t} \Big) \quad \text{for $\mathsf{p}< x \le 1$},\label{eq:phiplus}\\
\varphi^{-}(x,t) = e^{-\Theta_B(\mathsf{p}-x) t}\cdot {\varrho(x)}^{-} +
\frac{\Lambda_B(\mathsf{p}-x)}{\Theta_B(\mathsf{p}-x)} \cdot \Big(1- e^{-\Theta_B(\mathsf{p}-x)t}\Big) \quad \text{for $0 \le x < \mathsf{p}$}. \label{eq:phiminus}
\end{eqnarray}
When $n$ is large, we obtain from Theorem~\ref{lem:ode} the following approximation for the transient behavior of the shape of an order book: for interval $I \subset [0,1]$ and $t \in [0,T]$
\begin{eqnarray}
\zeta_t^{n,+}(I) \buildrel \Delta \over = \frac{1}{n^{2}}  \sum_{i> p_B^n(nt), \frac{i}{n}\in I } {\mathcal{X}^n_{i}} ( nt)
\approx  \zeta_t^{+}(I) = \int_{I} \varphi^{+}( x,t)dx  \approx  \frac{1}{n}  \sum_{\frac{i}{n}\in I } \varphi^{+}( \frac{i}{n},t),  \label{eq:selllimit}\\
\zeta_t^{n,-}(I) \buildrel \Delta \over = \frac{1}{n^{2}}  \sum_{i < p_A^n(nt), \frac{i}{n}\in I } {-\mathcal{X}^n_{i}} ( nt)
\approx  \zeta_t^{-}(I) = \int_{I} \varphi^{-}( x,t)dx  \approx  \frac{1}{n}  \sum_{\frac{i}{n}\in I } \varphi^{-}( \frac{i}{n},t),
 \label{eq:buylimit}
\end{eqnarray}
where $\varphi^{+}$ and $\varphi^{+}$ are given in \eqref{eq:phiplus} and \eqref{eq:phiminus}. In the next section, we will test the approximations \eqref{eq:selllimit} and \eqref{eq:buylimit} on data.

\section{Empirical Analysis}\label{sec:booktest}

In this section we perform in--sample analysis to
illustrate the potential relevance of our fluid approximation in Theorem~\ref{lem:ode}.



Our empirical analysis is based on one-month
message-level order book data from NYSE Arca (in short, Arca) in August 2010. As of 2009, around 20\% of market share for
NASDAQ-listed securities and 10\% of NYSE-listed securities are traded on Arca.

Two data sets are used. The first data set consists of all limit order activities on Arca, including limit order submission, modification and deletion. For each message in the data, it contains a time stamp down to
millisecond, the price and order size, the buy or sell indicator,
stock symbol, exchange, and an ID (identifier).  This data set enables us to recreate the limit order book at any give time for stocks traded on Arca. Our second data set records all the trades. Each message contains a time stamp down to
second, the traded price and order size, the buy or sell indicator, the best bid/ask prices and standing limit order quantities on these two prices when trades occur, stock symbol, and an ID (identifier). In conjunction with the first data set, we can
 analyze the detailed order flow properties such as limit order arrival rate, market order arrival rate and limit order cancellation rate.

%

\subsection{{Empirical analysis}} \label{subsec:bookestimate}
The goal of this section is to illustrate that our theoretical model can be potentially used to approximate the evolution of the shape of an order book on time scale of minutes for stocks and time periods which satisfy our model assumptions.

Our approach is to empirically test the approximations in \eqref{eq:selllimit} and \eqref{eq:buylimit}. For convenience we multiply both sides of these approximate equations by $n^2$, i.e., we test the following approximations and show that they are reasonably good:
\begin{eqnarray}
 \sum_{i> p_B^n(nt), \frac{i}{n}\in I } {\mathcal{X}^n_{i}} ( nt)
 \approx n  \sum_{\frac{i}{n}\in I } \varphi^{+}( \frac{i}{n},t),  \label{eq:selllimitn2}\\
  \sum_{i < p_A^n(nt), \frac{i}{n}\in I } {-\mathcal{X}^n_{i}} ( nt)
  \approx n  \sum_{\frac{i}{n}\in I } \varphi^{-}( \frac{i}{n},t).
 \label{eq:buylimitn2}
\end{eqnarray}
The left-hand side of \eqref{eq:selllimitn2} (respectively \eqref{eq:buylimitn2}) represents the total number of limit sell (respectively buy) orders in price interval $I \subset [0,1]$. With an appropriate shift of the price levels, such quantities can be directly obtained from data since we have the full empirical order book information at each time instant. On the other hand, the quantities at the right-hand side involve the values of functions $\varphi^+$ and $\varphi^-$ at discrete price levels $i/n$. We next discuss how to compute such values.

First, we screen the data to identify stocks and time periods which satisfy Assumptions \ref{ass:initialcondition} and \ref{ass:rateconverg1}. We find that liquidly-traded stocks, such as Bank of America Corp. (BAC), General Electric Co. (GE), Ford Motor Co. (F) and Dell Inc. (DELL), typically have large volumes of limit orders sitting at or near the best quotes and waiting
for market transactions. This is consistent with Assumption~\ref{ass:initialcondition}. Once a stock is chosen, we scan the intra-day order book data and look for time horizon (e.g., a few minutes) with low price volatility such that Assumption~\ref{ass:rateconverg1} is also satisfied.
For illustration purposes, we focus on the stock BAC during 12:45pm-13:05pm on August 5, 2010 as a representative example. The best bid and ask prices of BAC barely change during this 20-minute time period. 

Second, we estimate and input the needed parameters to compute $\varphi^+$ and $\varphi^-$ given in \eqref{eq:phiplus} and \eqref{eq:phiminus}. Such parameters include the initial order book profile $\varrho$, the number $\mathsf{p}$, the time length $nt$, the limit order arrival rates and the limit order {cancelation} rates. We {follow the four-step procedure described below to choose and compute these} parameters.
\begin{enumerate}
\item Normalize the calender time 12:45pm to be time zero. The initial order book is given by the first snapshot of the order book after 12:45pm. Hence we obtain $\varrho$ at each discrete price level.
\item Set $\mathsf{p}$ to be the best bid price at 12:45pm when computing $\varphi^+$ and the best ask when computing $\varphi^-$. Best bid and ask prices at 12:45pm are \$13.99 and \$14, respectively. 
\item Take snapshots of the order book every five minutes till 13:05pm. More precisely, we record the first snapshot of the order book after 12:50pm, 12:55pm, 13pm and 13:05pm, respectively.
We do so because we are interested in the dynamics of the shape on the `large' time scale in the order of minutes.

\item For each of the four time intervals: 12:45pm-12:50pm, 12:45pm-12:55pm, 12:45pm-13:00pm and 12:45pm-13:05pm, estimate the order flow intensities following~\cite{Cont10}. 
\end{enumerate}

\red{Finally, we compute the right-hand side of \eqref{eq:selllimitn2} and \eqref{eq:buylimitn2} by plugging in the model parameters. We take interval $I$ in \eqref{eq:selllimitn2} and \eqref{eq:buylimitn2} as 3-cent price bins, e.g., $[\$0.5, \$0.53]$.  It turns out that varying the bin size yields similar results, so we suppress further details.}



We now present our main empirical findings on in-sample tests of the approximations in \eqref{eq:selllimitn2} and \eqref{eq:buylimitn2}. We use Figure~1 for an illustration.
For detailed statistics and results on both buy-side and sell-side shapes, see the internet supplement available at the author's website \url{http://www1.se.cuhk.edu.hk/~xfgao/lobSupplement}.
Figure~1 plots the comparisons of the empirical sell-side shape of stock BAC and the theoretical sell-side shape obtained from applying our model at four different time instants. For the clarity of the presentation, we focus on price levels 10 ticks above the best ask price. This translates to seven `consecutive' 3--cent price bins in Figure~1.
For convenience, the aggregated order volume in each price bin is scaled by 363-shares, which is the average limit sell order size of BAC during 12:45pm-13:05pm. From Figure~1, we observe good agreement between the theoretical shapes computed from our model (Theorem~\ref{lem:ode}) and the empirical shapes of stock BAC on Arca. Thus, despite the simplifying assumptions, there is positive evidence that our limiting model can potentially explain the evolution of the shape of an order book over a time horizon that is large compared with the length of time between order book events.


Finally, we remark that the main goal of our empirical analysis is to illustrate the potential relevance of our fluid approximation rather than to make prediction about the future order--book shape, so no out--of--sample prediction result is reported here. In fact, we have attempted the analysis and found that the model does not work well for prediction purpose. One possible reason is that the parameters (order flow intensities) in our fluid model are constants, while actually they vary over time and are random on the time scale of several minutes.
This begs for further research on predictive models of order flow intensities and order--book shape, which is outside the scope of this paper. \\

\vspace{0.2cm}
\begin{figure}[H]
 \caption{Comparisons of empirical sell-side shapes and theoretical sell-side shapes for BAC on August 5, 2010. Each blue bar corresponds to the total number of limit sell orders in each 3-cent price bin and it is obtained from data. Each magenta bar corresponds to the total number of limit sell orders in each price bin and it is obtained from our model. The best ask prices are the same (\$14) at four time instants 12:50pm, 12:55pm, 13:00pm and 13:05pm. Each time instant (say, 12:50pm) represents the time that the first order book event occurs after that time instant.}
 \begin{tikzpicture}[x=0.5\linewidth,y=-0.38\linewidth]
    \path [use as bounding box] (-0.5,-0.5) rectangle(1.7,1.7);
    \path
    (0,0) node{\includegraphics[width=0.49\linewidth]{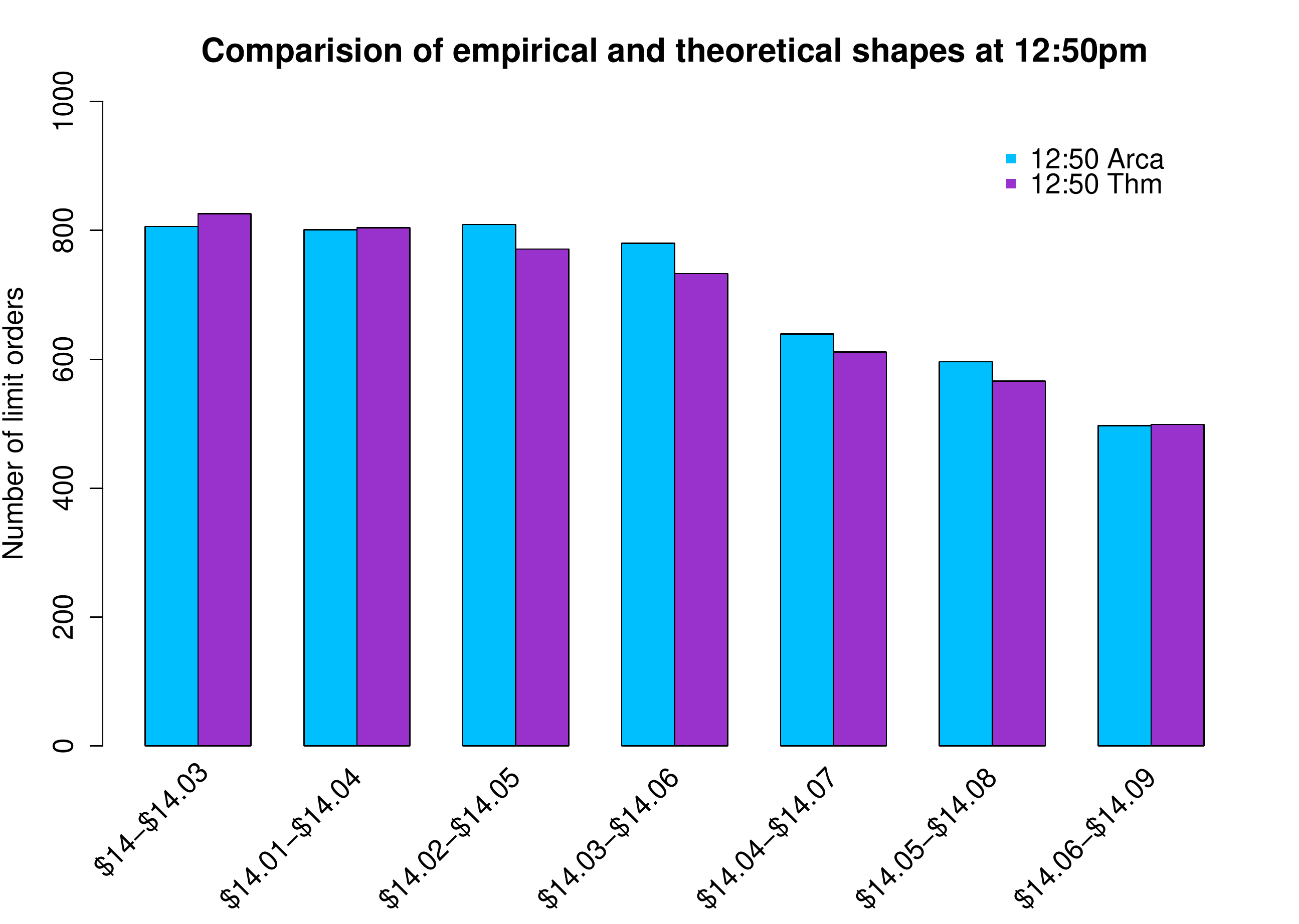}}
    (1,0) node{\includegraphics[width=0.49\linewidth]{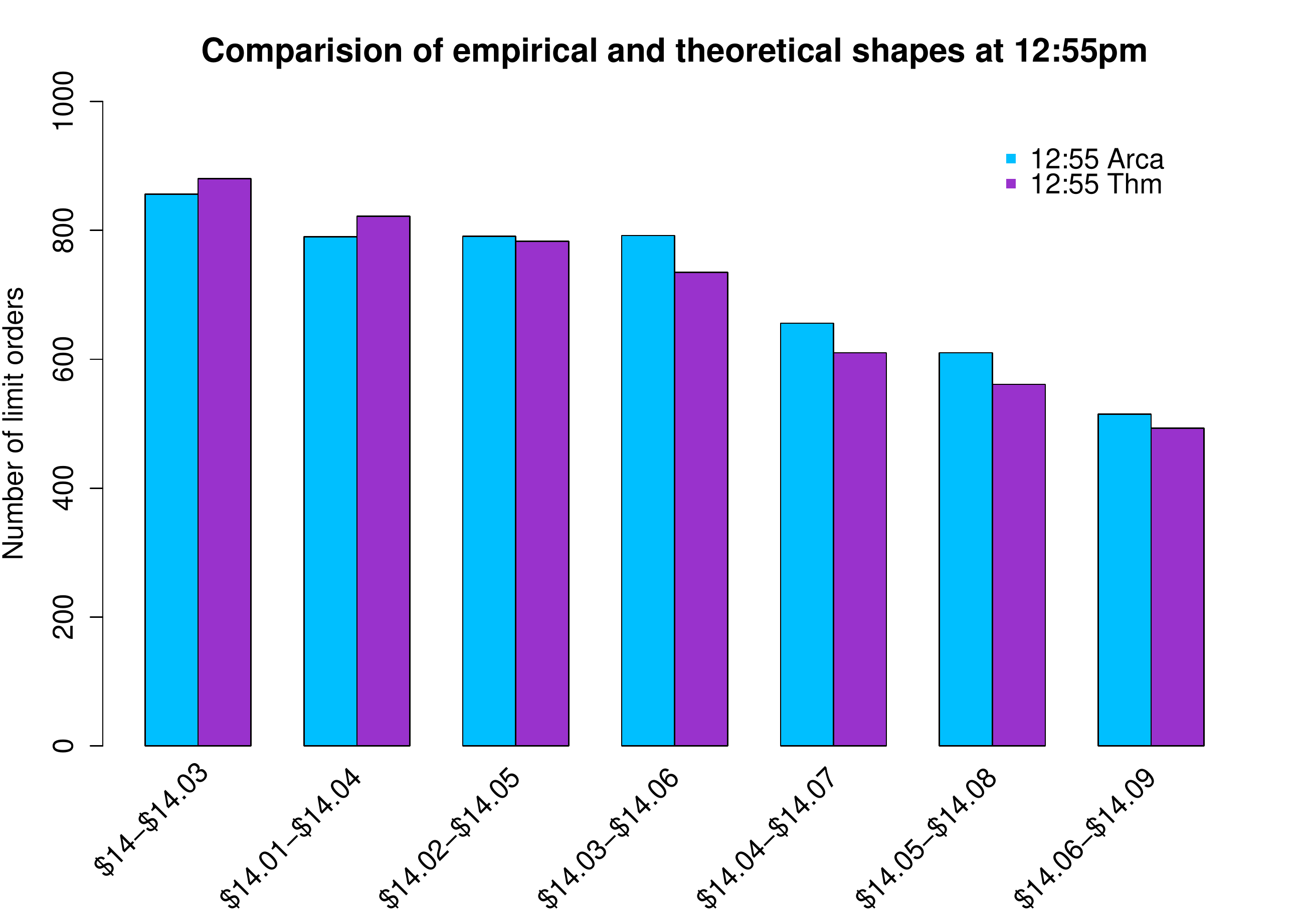}}

    (0,1) node{\includegraphics[width=0.49\linewidth]{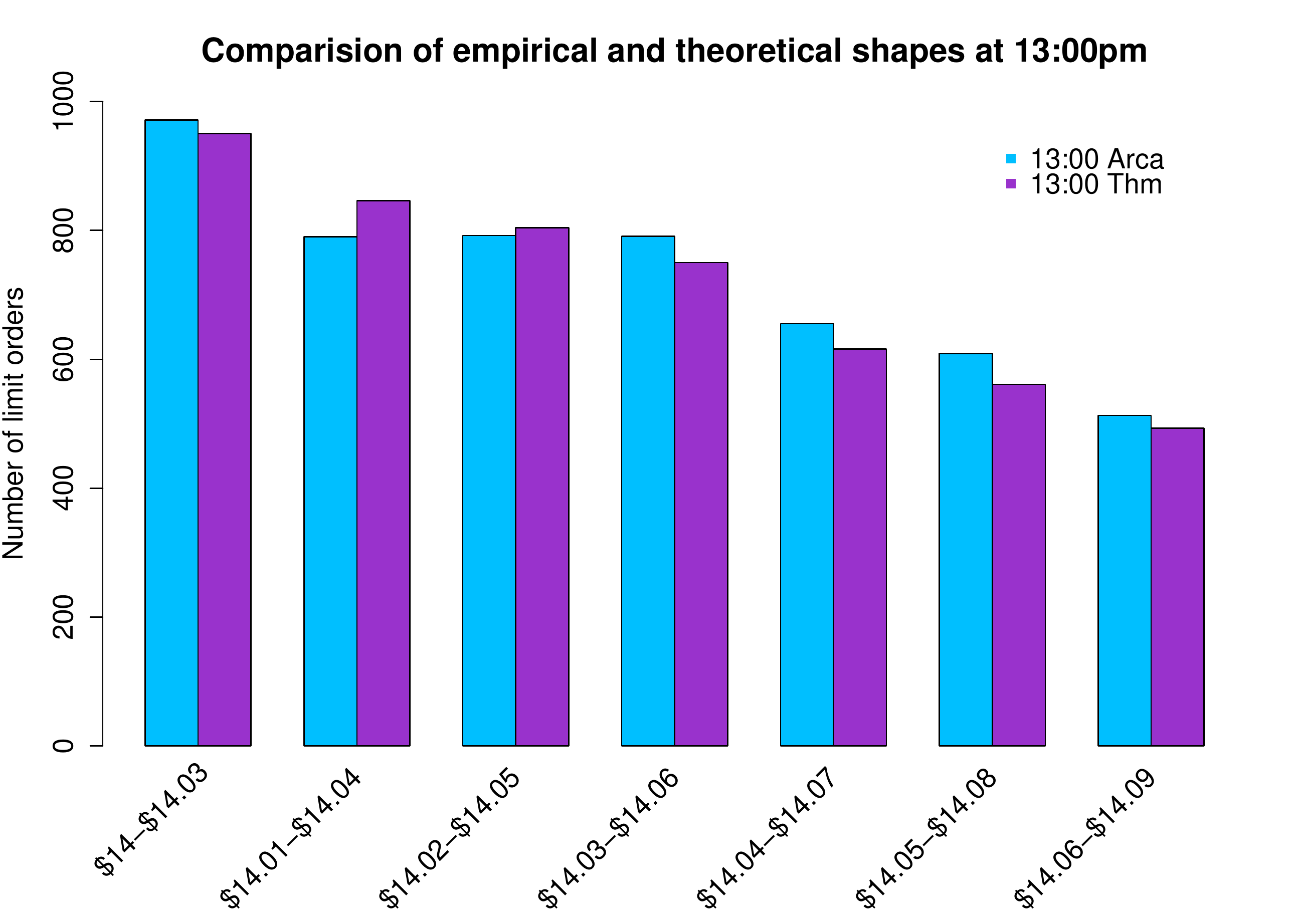}}
    (1,1) node{\includegraphics[width=0.49\linewidth]{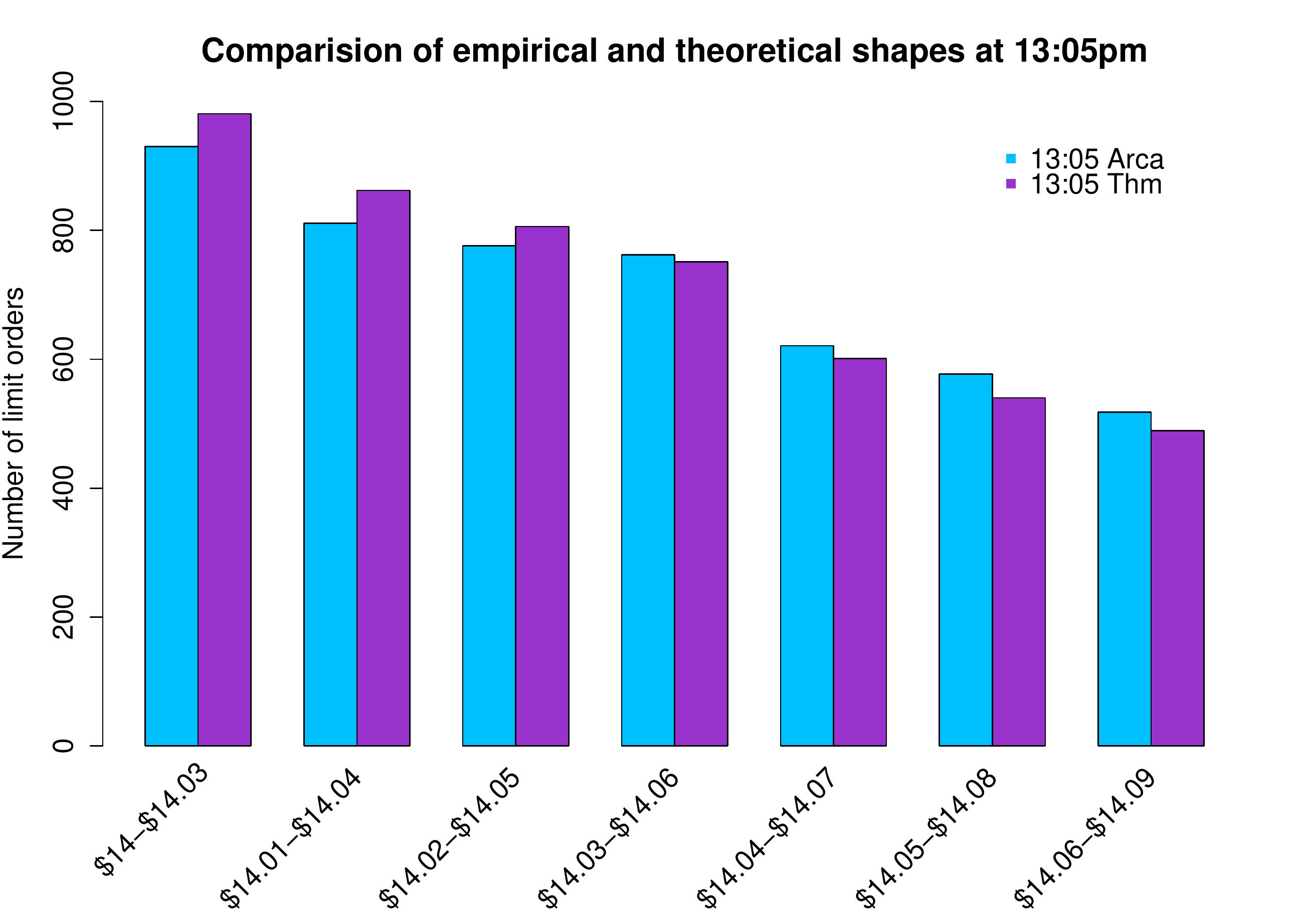}};
  \end{tikzpicture}
\end{figure}\label{fig:emp}

\newpage

\section{Preliminaries}\label{sec:prelim}
This section presents preliminaries for proving the main result (Theorem~\ref{lem:ode}).

\subsection{A technical lemma} \label{sec:tech-lemma}
In this section we introduce a technical lemma. It connects the process
$(\zeta^{n,+}, \zeta^{n,-})$ with the process $\zeta^{n}$ which plays a fundamental role in the proof of Theorem~\ref{lem:ode}.

Note that for fixed $f \in C[0,1]$, one can construct a family of functions $f_{\gamma} \in C[0,1]$ indexed by $\gamma \in (0, \mathsf{p})$ such that $\{f_{\gamma}\}$ are uniformly bounded and
\begin{equation} \label{eq:fgamma}
f_{\gamma}(x) =
\begin{cases} f(x)& \text{if $\mathsf{p} \le x \le 1 $,}
\\
0 &\text{if $0 \le x\le \mathsf{p}-\gamma$,}\\
\text{smooth} &\text{if $ \mathsf{p}-\gamma< x <\mathsf{p}$}.
\end{cases}
\end{equation}
It is clear that
\begin{eqnarray} \label{eq:fgamma-limit}
\lim_{\gamma \rightarrow 0+} f_{\gamma} (x) = f(x) \cdot 1_{[\mathsf{p},1]}(x)  \quad \text{for each $x\in[0,1]$}.
\end{eqnarray}

The following lemma says that $\langle \zeta^{n,+}_{t}, f \rangle$ is ``close" to $\langle \zeta^{n}_{t}, f_{\gamma} \rangle$ and $\langle \zeta^{n,-}_{t}, f \rangle$ is ``close" to $\langle \zeta^{n}_{\tau}, f_{\gamma} -f \rangle$ when $n$ is large and $\gamma$ is small.
The proof is given in Appendix~\ref{sec:lem-boundEdiff}.

\begin{lemma} \label{lem:boundEdiff}
\red{Fix a function $f\in C[0, 1]$.} There exists a positive constant $C$ \red{depending on $T$ and two nonnegative sequences $\{\tilde a_n\}, \{\hat a_n\}$} both depending on $T$ and both approaching 0 as $n \rightarrow \infty$, such that for each $\gamma \in (0, \mathsf{p})$ and $n$ large, we have
\begin{eqnarray}
{\E \Big[ \sup_{ 0 \le t \le T} \left|  \langle \zeta^{n,+}_{t}, f \rangle -  \langle \zeta^{n}_{t}, f_{\gamma} \rangle \right| \Big]} \le \tilde a_n +  C \gamma,\label{eq:bounddiffresult} \\
{\E \Big[ \sup_{0 \le t \le T} \left|  \langle \zeta^{n,-}_{t}, f \rangle -  \langle \zeta^{n}_{t}, f_{\gamma} -f \rangle \right| \Big]} \le  \hat a_n +  C \gamma. \label{eq:bounddiffresult2}
\end{eqnarray}
\end{lemma}

\subsection{Tightness of $\{(\zeta^{n,+}, \zeta^{n,-}): n \ge 1\}$ } \label{sec:tight}


In this section, we define the Polish space $\D([0, T], {\mathcal{\overline {M}}}([0,1]) )$ on which the pair of measure-valued processes $\{(\zeta^{n,+}, \zeta^{n,-}): n \ge 1\}$ lives and discuss the tightness of this sequence.


Recall that the space of finite {\it{non-negative}} measures on $[0,1]$, denoted by $\mathcal{M^+}([0,1])$, is a Polish space under the following metric $d_+$:
\begin{equation}
d_+ \left(\upsilon_{\alpha} , \upsilon_{\beta}\right)= \sum_{k=1}^\infty \frac{1}{2^k} \frac{|\langle \upsilon_{\alpha}, \phi_k \rangle -\langle \upsilon_{\beta}, \phi_k \rangle|}
{1+|\langle \upsilon_{\alpha}, \phi_k \rangle -\langle \upsilon_{\beta}, \phi_k \rangle| }
\label{eq:metricplus}
\end{equation}
where $\upsilon_{\alpha}, \upsilon_{\beta} \in \mathcal{M^+}([0,1])$, $\{\phi_k: k \ge 1\}$ are chosen to be a dense subset of ${C}([0,1])$ endowed with uniform topology,
and $
\langle \upsilon, \phi_k \rangle \equiv \int \phi_k(x) \upsilon(dx)$ for measure $\upsilon$, see, e.g., \cite[Section~4.1]{kipnis1999scaling}. Relying on the metric $d_+$, we can define a metric $d$ on the product space $\mathcal{M^+}([0,1]) \times \mathcal{M^+}([0,1])$ such that
\begin{equation}
d \left( (\upsilon^+_{\alpha}, \upsilon^-_{\alpha} ) , ( \upsilon^+_{\beta}, \upsilon^-_{\beta}) \right) \buildrel \Delta \over = \sqrt{  d_+^2 (\upsilon^+_{\alpha}, \upsilon^+_{\beta}) +  d_+^2 (\upsilon^-_{\alpha}, \upsilon^-_{\beta}) },
\label{eq:metric}
\end{equation}
where $(\upsilon^+_{\alpha}, \upsilon^-_{\alpha} ), ( \upsilon^+_{\beta}, \upsilon^-_{\beta}) \in \mathcal{M^+}([0,1]) \times \mathcal{M^+}([0,1])$.
The product space $\mathcal{M^+}([0,1]) \times \mathcal{M^+}([0,1])$ equipped with metric $d$ defined in (\ref{eq:metric}) is denoted by $\overline{\mathcal{M}}([0,1])$. One readily verifies that $\overline{\mathcal{M}}([0,1])$ is a Polish space (see, e.g. \cite{kotelenez2010stochastic}).
{Therefore, the Skorohod space $\mathbb{D}([0, T], \overline{\mathcal{M}}([0,1]) )$ on which $(\zeta^{n,+}, \zeta^{n,-})$ lives is also a Polish space with Skorohod $J_1$ topology. See \cite{bil99} for more details on the topology.}

\red{We now state the main result of this section. The proof is given in Appendix~\ref{sec:appendixA}. Recall a sequence of random elements is called $C$--tight if it is tight, and if all the limit points of the sequence are concentrated on continuous paths. See, e.g., \cite[Definition~VI.3.25]{Shiryaev03}.}


\begin{proposition} \label{lem:tightposneg0}
\red{Fix a function $f \in C[0,1]$. Then $\left\{ \left(\langle \zeta^{n,+}, f \rangle , \langle \zeta^{n,-}, f \rangle \right) : n \ge 1 \right\}$ is $C$--tight in $\mathbb{D}([0,T], \R^2)$, and $\{(\zeta^{n,+}, \zeta^{n,-}) : n \ge 1\}$ is $C$--tight in $\D([0, T], \overline{\mathcal{M}}([0,1]) )$.   }
\end{proposition}




Proposition~\ref{lem:tightposneg0} together with Prohorov's Theorem implies that the sequence $\{(\zeta^{n, +},  \zeta^{n, -}): n \ge 1\}$ is relatively compact. Our next step is to characterize the
limit points of this relatively
compact sequence. Suppose $(\zeta^+, \zeta^-)$ is a limit point, i.e., there is a subsequence
$\{(\zeta^{n_k, +},  \zeta^{n_k, -}): k=1, 2 \ldots\}$ such that
\begin{eqnarray} \label{eq:pinconvgpi}
(\zeta^{n_k, +},  \zeta^{n_k, -}) \Rightarrow (\zeta^+, \zeta^-) \quad \text{in $\D([0, T], \overline{\mathcal{M}}([0,1]) )$ \quad as $n_k \rightarrow \infty$}.
\end{eqnarray}
To characterize the pair $(\zeta^+, \zeta^-),$ we rely on an auxiliary process defined by
\begin{eqnarray}\label{eq:pit}
\zeta_t \buildrel \Delta \over = \zeta^+_{t} - \zeta^-_{t} \quad \text{for each $t \in [0,T]$},
\end{eqnarray}


\subsection{Characterization of the process $\zeta$}
In this section, we characterize the process $\zeta$ defined in \eqref{eq:pit}.

It is straightforward to verify that for $f \in C[0,1]$, the function from $(\pi^+, \pi^-) \in\D([0, T], \overline{\mathcal{M}}([0,1]) )$ to $\left(\langle \pi^+, f \rangle, \langle \pi^-, f \rangle \right)\in\D([0, T], \mathbb{R}^2 )$ is continuous with respect to Skorohod $J_1$ topology (use (\ref{eq:metricplus})-(\ref{eq:metric}) and Lemma~A.24 in \cite{seppalainen2008translation}).
Now \eqref{eq:pinconvgpi} together with continuous mapping theorem implies that for an arbitrary fixed function $f \in C([0,1])$ we have as $n_k \rightarrow \infty,$
\begin{eqnarray} \label{eq:pinconvgpi-f}
\left(\langle \zeta^{n_k,+}, f \rangle, \langle \zeta^{n_k,-}, f \rangle \right) \Rightarrow \left( \langle \zeta^{+}, f \rangle,  \langle \zeta^{-}, f \rangle \right) \quad \text{in $\D([0, T], \mathbb{R}^2)$}.
\end{eqnarray}
Thus we deduce that
\begin{eqnarray} \label{eq:convgzetaf}
\langle \zeta^{n_k}, f \rangle \Rightarrow \langle \zeta, f \rangle \quad \text{in $\D([0, T], \mathbb{R})$,} \end{eqnarray}
where $\zeta_{}^{n}$ is given in \eqref{eq:empirialshape} and $\zeta$ is given in \eqref{eq:pit}.

Next, for fixed $n \ge 1$,
if we define a function $F$ on $\mathbb{R}^n$ by
\begin{equation} \label{eq:F}
F(x_1,\ldots,x_n)=\frac{1}{n^{2}}\sum_{k=1}^n f\Big(\frac{k}{n}\Big)x_k,
\end{equation}
we obtain from $(\ref{eq:empirialshape})$ and Markov property of $\mathcal{X}^n$ that for $t \in [0,T]$,
\begin{eqnarray} \label{eq:repzetatnf}
\langle \zeta_t^n, f \rangle
& =&\frac{1}{n^{2}} \sum_{i=1}^{n} \mathcal{X}^n_{i} (nt )
\cdot f \Big( \frac{i}{n}\Big) = F( \mathcal{X}^n (nt) ) \nonumber\\
&=& \langle \zeta_0^n, f \rangle + \int_0^{nt}
\mathcal{L}_n F(\mathcal{X}^n(s))ds+ \mathsf{M}^n_t.
\end{eqnarray}
Here $\mathcal{L}_n$ is operator given in (\ref{eq:generator}) and $\mathsf{M}^n$ is a (local) martingale. The martingale representation \eqref{eq:repzetatnf} and the convergence in \eqref{eq:convgzetaf} are the keys to characterize $\zeta$.

To state the main result of this section,
we recall the functions $\Lambda_A, \Lambda_B, \Theta_A$ and $\Theta_B$ given in Assumption~\ref{ass:rateconverg1} and $\mathsf{p}$ given in Assumption~\ref{ass:initialcondition}. For notational convenience, we let
$\nu_\Lambda$ be a signed measure that is absolutely continuous with
respect to Lebesgue measure such that for $x \in [0,1]$
\begin{eqnarray} \label{eq:nuLambda}
\frac{\nu(dx)}{dx}= \Lambda_A(x- \mathsf{p}) \cdot 1_{x> \mathsf{p}}  - \Lambda_B( \mathsf{p}-x) \cdot 1_{x< \mathsf{p}}.
\end{eqnarray}
In addition, we let $\mathcal{A}_{\Theta}$ be a linear operator such that for $f \in C([0,1])$,
\begin{eqnarray} \label{eq:ATheta}
\mathcal{A}_{\Theta}f(x)=f(x) \cdot \Big( \Theta_A(x- \mathsf{p}) \cdot 1_{x> \mathsf{p}}  + \Theta_B( \mathsf{p}-x) \cdot 1_{x< \mathsf{p}}\Big).
\end{eqnarray}

We now state the main result of this section which gives a characterization of $\zeta$.
\begin{proposition} \label{lem:limitpointpi}
{Let $(\zeta^+, \zeta^-)$ be a limit point in \eqref{eq:pinconvgpi} and }
let $\zeta$ be defined in (\ref{eq:pit}). Then $\zeta$ satisfies the following set of equations: for any $f \in C([0,1])$ and $t \in [0, T]$,
\begin{equation} \label{eq:generatormart3}
\langle \zeta_t, f \rangle = \langle \zeta_0, f \rangle + \langle
\nu, f \rangle \cdot t - \int_0^t \langle \zeta_s,
\mathcal{A}_{\Theta}f \rangle ds,
\end{equation}
where $\nu$ and $\mathcal{A}_{\Theta}$ are given in (\ref{eq:nuLambda}) and (\ref{eq:ATheta}). \red{Moreover, $\langle \zeta, f \rangle$ has continuous trajectory and there is a unique deterministic measure-valued process solving Equations \eqref{eq:generatormart3}.} Finally, for fixed $t\in[0,T]$, $\zeta_t$ is a finite measure that is absolutely continuous with respect to Lebesgue
measure, and it has bounded density function $\varphi$ determined by \eqref{eq:varphiinitial}--\eqref{eq:varphip}.
%
%
\end{proposition}  
{The main technical part of the proof of Proposition~\ref{lem:limitpointpi} consists of showing that in \eqref{eq:repzetatnf}
 as $n\rightarrow \infty$, the martingale term vanishes and the term involving the generator converges weakly.
We leave the details of the proof to Appendix~\ref{app:2}.}


\section{Proof of Theorem~\ref{lem:ode}} \label{sec:proofthm1}
\add{In this section we prove our main result Theorem~\ref{lem:ode}. The proof of part (b) of Theorem~\ref{lem:ode} relies on the results in Section~\ref{sec:prelim} while the proof of part (a) does not.}

\subsection{Proof of part (a) of Theorem~\ref{lem:ode}}

\begin{proof}
We prove (\ref{eq:askpriceconvg}) using a stochastic comparison method. The convergence in (\ref{eq:bidpriceconvg}) follows from a similar
argument and {is thus} omitted.

We first note that
\begin{eqnarray} \label{eq:initialpriceconvg}
\lim_{n \rightarrow \infty} \frac{p_A^n(0)}{n} = \lim_{n \rightarrow
\infty} \frac{p_B^n(0)}{n} = \mathsf{p}.
\end{eqnarray}
This readily follows from the definition of best bid
and best ask prices and the regularity condition (b) of $\varrho$ in Assumption~\ref{ass:initialcondition}.
Hence in order to prove (\ref{eq:askpriceconvg}), it suffices to show
\begin{eqnarray*}
 \sup_{0 \le t \le T} \left| \frac{p_A^n(nt)}{n}- \frac{p_A^n(0)}{n} \right| \Rightarrow 0, \quad \text{as $n \rightarrow \infty$}. \label{eq:askpriceconvg'}
\end{eqnarray*}
Given any $\epsilon>0$, we have
\begin{eqnarray} \label{eq:askpriceconvg2parts}
\lefteqn{{\mathbb{P} \Big(  \sup_{0 \le t \le T} \left|
\frac{p_A^n(nt)}{n}-
\frac{p_A^n(0)}{n} \right|  > \epsilon \Big)}} \nonumber \\
 & \le \mathbb{P}\Big(  \mathop {\sup }\limits_{0 \le t \le T} {p_A^n(nt)}- {p_A^n(0)}
> n \epsilon \Big) +  \mathbb{P}\Big( {p_A^n(0)} - \mathop {\inf }\limits_{0 \le t \le T} {p_A^n(nt)}
> n \epsilon \Big).
\end{eqnarray}

We next show that for any small $\delta>0,$ there exists
$N_{\delta}$ such that when $n > N_{\delta},$
\begin{eqnarray} \label{eq:askpriceconvgpart1}
\mathbb{P} \Big(  \sup_{0 \le t \le T} {p_A^n(nt)}- {p_A^n(0)}
> n \epsilon \Big)
\le \delta.
\end{eqnarray}
Since $p_A^n(\cdot)$ is upper bounded by $n+1$ by definition, we deduce that if $p_A^n(0) \ge n +1 -  n \epsilon $, then \eqref{eq:askpriceconvgpart1} follows trivially. Thus below we assume $p_A^n(0) < n +1 -  n \epsilon$, or equivalently,
$p_A^n(0) + \left \lfloor n \epsilon \right
\rfloor \le n$.
For {a fixed $n$ and a given } ${p_A^n(0)}$, we first define a {map} $g$ from $\mathbb{Z}^n$ to $\mathbb{Z}$:
\begin{eqnarray} \label{eq:g}
g(x) = \sum_{i=p_A^n(0)}^{p_A^n(0) + \left \lfloor n \epsilon \right
\rfloor} x_i^+ .
\end{eqnarray}
We use the {map} $g$ to further introduce an auxiliary process $\mathcal{Z}^n$ to track the
number of sell limit orders on price levels from $p_A^n(0)$ to
$p_A^n(0) + \left \lfloor n \epsilon \right \rfloor$. That is, for each
$t \ge 0,$
\begin{eqnarray} \label{eq:Znt}
\mathcal{Z}^n(t) \buildrel \Delta \over = g(\mathcal{X}^n(nt) ) = \sum_{i=p_A^n(0)}^{p_A^n(0) + \left \lfloor n \epsilon \right
\rfloor} [\mathcal{X}^n_i(nt)]^+.
\end{eqnarray}
Write
\begin{equation}
\sigma^n_{\mathcal{Z}^n} \buildrel \Delta \over = \inf \{ t \ge 0 : \mathcal{Z}^n(t) = 0\}.
\end{equation}
Then $\sigma^n_{\mathcal{Z}^n}$ is
the first time that $\mathcal{Z}^n$ reaches $0$ starting from $\mathcal{Z}^n(0)>0$, where $\mathcal{Z}^n(0)>0$ follows from the
fact that $p_A^n(0)\le n$. Thus in order to show (\ref{eq:askpriceconvgpart1}), it suffices to prove
\begin{eqnarray} \label{eq:askpriceconvgpart1'}
\mathbb{P}( \sigma^n_{\mathcal{Z}^n} \le T) \le \delta, \quad \text{when $n>N_{\delta}$.}
\end{eqnarray}

The key idea for proving \eqref{eq:askpriceconvgpart1'} is to construct a $n-$dimensional CTMC $\{\mathcal{W}^n(t): t \ge 0 \}$ on $\mathbb{Z}^n$ that is dominated by $\{\mathcal{X}^n(t): t \ge 0 \}$ in the sense of strong stochastic order for each $n$. That is, $\E r({\mathcal{W}^n(t)}) \le \E r({\mathcal{X}^n(t)})$ for all real-valued increasing functions $r$ on $\mathbb{Z}^n$ and all $t\ge0$. 
To proceed, we set ${\mathcal{W}^n}(0) = {\mathcal{X}^n}(0)$. For convenience, we define two constants:
\begin{eqnarray}
\bar{\Lambda} &\buildrel \Delta \over=& 2 \cdot [ (\max_{x \in [0,1]} \Lambda_A (x)) \vee  (\max_{x \in [0,1]} \Lambda_B (x))] < \infty, \label{eq:lambda-max} \\
\bar \Theta &\buildrel \Delta \over=& (\max_{x \in [0,1]} \Theta_A(x))
\vee (\max_{x \in [0,1]} \Theta_B(x)) >0.\label{eq:theta-max}
\end{eqnarray}
We construct ${\mathcal{W}^n}$ by modifying the transition parameters of the CTMC ${\mathcal{X}^n}$ as follows: (a) set the limit sell order arrival rate to zero; (b) set the limit buy order cancellation rate and the market sell order arrival rate to zero; (c) set the limit buy order arrival rate at each price level to $\bar {\Lambda}$ given in \eqref{eq:lambda-max}, the market buy order arrival rate to $n^{ \kappa} \cdot \Upsilon$ given in Assumption~\ref{ass:rateconverg1}, and the limit sell order cancellation rate per order to $\frac{1}{n} \bar {\Theta} $. 

We next argue ${\mathcal{W}^n}$ is stochastically dominated by ${\mathcal{X}^n}$. We apply the result of
\cite{massey1987stochastic}. 
 \red{Recall a set $\Gamma \subset \mathbb{Z}^n$ is said to be decreasing if $x \in \Gamma$ implies $\{y \in \mathbb{Z}^n: y \le x\} \subset \Gamma$. From the construction of ${\mathcal{W}^n}$, we know that there are no limit sell order arrivals and limit buy order removals (due to cancellation or matches with market sell orders) in ${\mathcal{W}^n}$. This implies that if ${\mathcal{W}^n}$ transits from one state $w \in \mathbb{Z}^n$ to another state $w'\in \mathbb{Z}^n$, then $w' \le w$ with respect to the partial order $\le$ on $\mathbb{Z}^n$.} Thus we deduce from \cite[Theorem~5.3]{massey1987stochastic} that it suffices to show for all decreasing sets $\Gamma \subset \mathbb{Z}^n$ such that $w \notin \Gamma$
\begin{eqnarray} \label{eq:order}
\sum_{z \in \Gamma} q_{\mathcal{W}^n}(w,z) \ge \sum_{z \in \Gamma} q_{\mathcal{X}^n}(x, z) \quad \text{for all $w \le x$,}
\end{eqnarray}
where $q_{\mathcal{W}^n}$ and $q_{\mathcal{X}^n}$ are transition rates for ${\mathcal{W}^n}$ and ${\mathcal{X}^n}$. One readily verifies \eqref{eq:order} from the construction of ${\mathcal{W}^n}$ and the fact that both ${\mathcal{W}^n}$ and ${\mathcal{X}^n}$ have the property that only one component can change its value at each transition.

Given that ${\mathcal{W}^n}$ is stochastically dominated by ${\mathcal{X}^n}$, we are ready to construct a process $\tilde{\mathcal{Z}}^n$ that is stochastically dominated by ${\mathcal{Z}^n}$ given in \eqref{eq:Znt}. We simply set for $t \ge 0$
\begin{eqnarray} \label{eq:tildeZn}
\tilde{\mathcal{Z}}^n (t) \buildrel \Delta \over = g(\mathcal{W}^n(nt)) = \sum_{i=p_A^n(0)}^{p_A^n(0) + \left \lfloor n \epsilon \right
\rfloor} [\mathcal{W}^n_i(nt)]^+ \ge 0.
\end{eqnarray}
Then ${\mathcal{Z}^n}$ dominates $\tilde{\mathcal{Z}}^n$ in strong stochastic order since $g$ in \eqref{eq:g} is an increasing function on $\mathbb{Z}^n$. We then deduce that
\begin{eqnarray*} \label{eq:compare}
 \mathbb{P}( \sigma^n_{\mathcal{Z}^n} \le T) \le  \mathbb{P}( \sigma^n_{\tilde{\mathcal{Z}}^n}\le T),
\end{eqnarray*}
where $\sigma^n_{\tilde{\mathcal{Z}}^n}$ is the first time that $\tilde{\mathcal{Z}}^n$ reaches zero
starting from ${\tilde{\mathcal{Z}}^n}(0) ={{\mathcal{Z}}^n}(0)$, i.e.,
\begin{equation}
\sigma^n_{\tilde{\mathcal{Z}}^n} \buildrel \Delta \over = \inf \{ t \ge 0 : \tilde{\mathcal{Z}}^n(t) = 0\}.
\end{equation}
Hence to show \eqref{eq:askpriceconvgpart1'} it suffices to prove that given
any $\delta>0,$ there exists $N_{\delta}$ such that
\begin{eqnarray} \label{eq:askpriceconvgpart1''}
 \mathbb{P}( \sigma^n_{\tilde{\mathcal{Z}}^n} \le T) \le \delta, \quad \text{when $n
> N_{\delta}$.}
\end{eqnarray}

We now focus on proving (\ref{eq:askpriceconvgpart1''}). \red{One can verify that
the process
$\tilde{\mathcal{Z}}^n$ defined in \eqref{eq:tildeZn} is a pure-death process with absorbing barrier at zero, 
and the death rate of
$\tilde{\mathcal{Z}}^n$ when it is in state $k \in \{1, 2,\ldots \}$ is given by
\begin{eqnarray} \label{eq:deathrate}
n \cdot \big(k \cdot \frac{\bar {\Theta}}{n} +  {n^{\kappa}} \cdot \Upsilon \big).
\end{eqnarray}}
%
%
One can also check from Assumption~\ref{ass:initialcondition}, (\ref{eq:tildeZn}) and (\ref{eq:initialpriceconvg}) that
\begin{eqnarray} \label{eq:Zn0lmt}
\lim_{n \rightarrow \infty} \frac{{\tilde{\mathcal{Z}}^n}(0)}{n^2} = \lim_{n \rightarrow \infty} \frac{1}{n^2} \sum_{i=p_A^n(0)}^{p_A^n(0) + \left \lfloor n \epsilon
\right \rfloor} [\mathcal{X}^n_i(0)]^+
 = \int_{\mathsf{p}}^{\mathsf{p} + \epsilon} \varrho(x)dx >0.
\end{eqnarray}
Since $\tilde {\mathcal{Z}}^n$ is a pure death process with absorbing state zero, we obtain
\begin{eqnarray} \label{eq:sigmatildeZn}
\sigma^n_{\tilde{\mathcal{Z}}^n} = \sum_{k=1}^{{\tilde{\mathcal{Z}}^n}(0)} D_k,
\end{eqnarray} 
where $D_k$ represents the first passage time that ${\tilde{\mathcal{Z}}^n}$ starts from state $k$ and reaches $k-1$, so $D_k$ is an exponential random variable with rate $k \cdot{\bar {\Theta}} + {n^{1+ \kappa}} \Upsilon$ given in \eqref{eq:deathrate}. In addition, all the $D_k$'s
are mutually independent. Thus we have
\begin{eqnarray*}
\E [ \sigma^n_{\tilde{\mathcal{Z}}^n}] &=& \E \Big[\sum_{k=1}^{{\tilde{\mathcal{Z}}^n}(0) } D_k \Big] = \sum_{k=1}^{{\tilde{\mathcal{Z}}^n}(0) }
\frac{1}{k \cdot{\bar {\Theta}} + {n^{1+ \kappa }} \cdot \Upsilon} , \\
Var( \sigma^n_{\tilde{\mathcal{Z}}^n} ) &=&  \sum_{k=1}^{{\tilde{\mathcal{Z}}^n}(0)} Var(D_k) =
\sum_{k=1}^{\tilde{\mathcal{Z}}^n(0)} \frac{1}{(k \cdot{\bar {\Theta}} + {n^{1+ \kappa}} \cdot \Upsilon)^2} . 
\end{eqnarray*}
Since $\kappa <1$, we deduce from \eqref{eq:Zn0lmt} that there exists $C$ independent of $n$ such that
\begin{eqnarray}
\lim_{n \rightarrow \infty} \frac{1}{ \ln n} \E [\sigma^n_{\tilde{\mathcal{Z}}^n}] = \frac{1-\kappa}{\bar{\Theta}} >0 , \quad \mbox{and} \quad \sup_{n }  Var( \sigma^n_{\tilde{\mathcal{Z}}^n}) &\le & C. \label{eq:2mmtsigmaZ}
\end{eqnarray} 
\red{Thus for all $n$ large enough, we obtain $\E [\sigma^n_{\tilde{\mathcal{Z}}^n}] > T$. It follows that
\begin{eqnarray*}
 Var( \sigma^n_{\tilde{\mathcal{Z}}^n}) &\ge& \E \Big[ \big( \E [\sigma^n_{\tilde{\mathcal{Z}}^n}] - \sigma^n_{\tilde{\mathcal{Z}}^n} \big)^2:  \sigma^n_{\tilde{\mathcal{Z}}^n} \le T \Big] \\
 &\ge & \E \Big[ \big( \E [\sigma^n_{\tilde{\mathcal{Z}}^n}] - T \big)^2:  \sigma^n_{\tilde{\mathcal{Z}}^n} \le T \Big]\\
 & = & \big( \E [\sigma^n_{\tilde{\mathcal{Z}}^n}] - T \big)^2 \cdot  \mathbb{P}( \sigma^n_{\tilde{\mathcal{Z}}^n} \le T).
\end{eqnarray*}}
Therefore we deduce that when $n$ is large,
\begin{eqnarray*}
 \mathbb{P}( \sigma^n_{\tilde{\mathcal{Z}}^n } \le T) &\le &  \frac{1}{ \Big( \E [\sigma^n_{\tilde{\mathcal{Z}}^n}] - T\Big)^2 } Var( \sigma^n_{\tilde{\mathcal{Z}}^n}).
\end{eqnarray*}
This yields (\ref{eq:askpriceconvgpart1''}) after applying \eqref{eq:2mmtsigmaZ}. Thus we have completed the proof of \eqref{eq:askpriceconvgpart1}.

We next show that for $n$ large,
\begin{eqnarray} \label{eq:askpriceconvgpart2}
\mathbb{P}\Big( {p_A^n(0)} - \mathop {\inf }\limits_{0 \le t \le T} {p_A^n(nt)}
> n \epsilon \Big) \le \delta.
\end{eqnarray}
Note that $p_B^n(\cdot)$ is smaller than $p_A^n(\cdot)$ at each time, we have
\begin{eqnarray*}
{p_A^n(0)} - \inf_{0 \le t \le T} {p_A^n(nt)} &\le& {p_A^n(0)} -
\inf_{0 \le t \le T} {p_B^n(nt)} \\
&=& {p_B^n(0)} - \inf_{0 \le t \le T} {p_B^n(nt)} + {p_A^n(0)} -
{p_B^n(0)}.
\end{eqnarray*}
Now (\ref{eq:initialpriceconvg}) implies that for $n$ large,
\begin{eqnarray*}
\mathbb{P} \Big( {p_A^n(0)} - {p_B^n(0)}  > \frac{n \epsilon}{2} \Big) \le
\frac{\delta}{2}.
\end{eqnarray*}
Thus it suffices to prove that for $n$ large
\begin{eqnarray*}
\mathbb{P} \Big( {p_B^n(0)} - \inf_{0 \le t \le T} {p_B^n(nt)} > \frac{n \epsilon}{2} \Big) \le
\frac{\delta}{2}.
\end{eqnarray*}
This follows from a similar argument for
(\ref{eq:askpriceconvgpart1}). Thus we have \eqref{eq:askpriceconvgpart2}. On combining (\ref{eq:askpriceconvg2parts}), (\ref{eq:askpriceconvgpart1}), and
(\ref{eq:askpriceconvgpart2}), we obtain (\ref{eq:askpriceconvg}). Therefore the proof is complete.
\end{proof}

\subsection{Proof of part (b) of Theorem~\ref{lem:ode}}
In this section we prove part (b) of Theorem~\ref{lem:ode}. We rely on the results in Section~\ref{sec:prelim} and characterize {the limit point} $(\zeta^+, \zeta^-)$ in \eqref{eq:pinconvgpi}.

\begin{proof}
{Assume $(\zeta^+, \zeta^-)$ is a limit point as in Proposition~\ref{lem:limitpointpi}.}
In view of Proposition~\ref{lem:limitpointpi}, to show that $\zeta^+_t$ (respectively $\zeta^-_t$) is absolutely continuous with density $\varphi^+(\cdot, t)$ (respectively $\varphi^-(\cdot, t)$),
it suffices to show that
 for any $f \in C([0,1])$ and $t \in [0, T]$, we have
\begin{eqnarray}
\langle \zeta^+_t, f \rangle =  \int_{0}^{1} f(x) \varphi^+(x,t) dx, \label{eq:zetaplusandminus1}\\
\langle \zeta^-_t, f \rangle =  \int_{0}^{1} f(x) \varphi^-(x,t) dx, \label{eq:zetaplusandminus2}
\end{eqnarray}
where $\varphi^+ = \max\{\varphi, 0\}$, $\varphi^- = \max\{-\varphi, 0\}$ and $\varphi$ is uniquely determined by \eqref{eq:varphiinitial}--\eqref{eq:varphip}.

\red{To this end, we first note that $\langle \zeta^{n_k, +}_{t}, f\rangle \Rightarrow \langle \zeta^+_{t}, f \rangle$ as $n_k \rightarrow \infty$ for fixed time $t \in [0, T]$. This is true because $\langle \zeta^{n_k, +}, f\rangle \Rightarrow \langle \zeta^+, f \rangle$ in \eqref{eq:pinconvgpi-f} and
almost every path of the limit point $\langle \zeta^+, f \rangle$ is continuous by Proposition~\ref{lem:tightposneg0}. Similarly, we deduce from \eqref{eq:convgzetaf} that
$\langle \zeta^{n_k}_{t}, f_{\gamma} \rangle \Rightarrow \langle \zeta_{t}, f_{\gamma} \rangle$ for $f_{\gamma} \in C[0,1]$ introduced in \eqref{eq:fgamma} and $t \in [0, T]$. Since $\zeta$ is deterministic by Proposition~\ref{lem:limitpointpi}, we obtain that $\langle \zeta^{n_k}_{t}, f_{\gamma} \rangle$ converges in probability to $\langle \zeta_{t}, f_{\gamma} \rangle$. Therefore, we have for fixed $\gamma>0$,}
\begin{eqnarray*}
\langle \zeta^{n_k, +}_{t}, f \rangle -  \langle \zeta^{n_k}_{t}, f_{\gamma} \rangle \Rightarrow  \langle \zeta^{+}_{t}, f \rangle -  \langle \zeta_{t}, f_{\gamma} \rangle \quad \text{as $n_k \rightarrow \infty$.}
\end{eqnarray*}
Using Theorem~3.4 in \cite{bil99}, we obtain from Lemma~\ref{lem:boundEdiff} that there exists a constant $C>0$ such that for each $\gamma \in (0, \mathsf{p})$ and $t \in [0, T]$
\begin{eqnarray} \label{eq:zetafbound}
{\E  \left|  \langle \zeta^{+}_{t}, f \rangle -  \langle \zeta_{t}, f_{\gamma} \rangle \right|} &\le&
\mathop {\lim\inf}_{n_k \rightarrow \infty}{\E  \left|  \langle \zeta^{n_k, +}_{t}, f \rangle -  \langle \zeta^{n_k}_{t}, f_{\gamma} \rangle \right|} \nonumber \\
&\le& \mathop {\liminf}_{n_k \rightarrow \infty}  (\tilde a_{n_k} +  C \gamma) \nonumber \\
& = &  C \gamma.
\end{eqnarray}
In addition, since $\zeta_t$ is a deterministic signed measure with bounded density function by Proposition~\ref{lem:limitpointpi} and $\{f_{\gamma}: \gamma \in (0, \mathsf{p}) \}$ is uniformly bounded on $[0,1]$, one infers from the
dominated convergence theorem and \eqref{eq:fgamma-limit} that
\begin{eqnarray}\label{eq:1}
\lim_{\gamma \rightarrow 0+} {  \langle \zeta_{t}, f_{\gamma} \rangle } = \langle \zeta_{t}, \lim_{\gamma \rightarrow 0+} f_{\gamma} \rangle
=  \int_{\mathsf{p}}^{1} f(x) \varphi(x,t) dx .
\end{eqnarray}
Since $\varphi$ satisfies \eqref{eq:varphiinitial}-\eqref{eq:varphip}, and $\rho$ satisfies \eqref{eq:gammap}, we deduce that
$\varphi$ is non-negative on $[\mathsf{p},1]$ and non-positive on $[0, \mathsf{p}]$. This implies
\begin{eqnarray}\label{eq:2}
\int_{\mathsf{p}}^{1} f(x) \varphi(x,t) dx = \int_{0}^{1} f(x) \varphi^+(x,t) dx.
\end{eqnarray}
Therefore, we find from \eqref{eq:zetafbound} that
\begin{eqnarray} \label{eq:zetafbound2}
{\E  \left|  \langle \zeta^{+}_{t}, f \rangle -  \int_{0}^{1} f(x) \varphi^+(x,t) dx \right|}
& \le &   {\E  \left|  \langle \zeta^{+}_{t}, f \rangle -  \langle \zeta_{t}, f_{\gamma} \rangle \right|} +
{  | \langle \zeta_{t}, f_{\gamma} \rangle - \int_{0}^{1} f(x) \varphi^+(x,t) dx |} \nonumber \\
& \le & C \gamma + {  | \langle \zeta_{t}, f_{\gamma} \rangle - \int_{0}^{1} f(x) \varphi^+(x,t) dx |} .
\end{eqnarray}
Let $\gamma \rightarrow 0+$ in inequality (\ref{eq:zetafbound2}), we conclude from \eqref{eq:1} and \eqref{eq:2} that
\begin{eqnarray*}
{\E  \left|  \langle \zeta^{+}_{t}, f \rangle - \int_{0}^{1} f(x) \varphi^+(x,t) dx \right|}= 0.
\end{eqnarray*}
Thus (\ref{eq:zetaplusandminus1}) holds with probability one. A similar argument yields (\ref{eq:zetaplusandminus2}). The proof is therefore complete.
\end{proof}

\section*{Acknowledgement}
We are greatful to Jim Dai, Ton Dieker, Rama Cont and Yu Gu
for helpful discussions.  Xuefeng Gao acknowledges support from ECS Grant 2191081, and
CUHK Direct Grants 4055035 and 4055054. S.J. Deng is supported in part by a PSERC grant.

\newpage

\appendix

\section{Proof of Lemma~\ref{lem:boundEdiff}} \label{sec:lem-boundEdiff}
In this section we prove Lemma~\ref{lem:boundEdiff} in Section~\ref{sec:tech-lemma}.
To do so, we first introduce a result on bounding the (scaled) total number of limit orders within a certain price range of the order book. The proof is deferred to the end of this appendix. 

\begin{lemma} \label{lem:totalorder}
Fix $T>0$. There exists a positive constant $C$ which depends on $T$
but is independent of $n$ such that for any fixed interval $[a,b] \subset [0,1]$, we have
\begin{eqnarray}
\mathop {\limsup}_{n \rightarrow \infty} \E \left[ \sup_{0 \le t \le T} \frac{1}{n^{2}}
\sum_{i: \frac{i}{n} \in [a,b]} |\mathcal{X}^n_{i} (nt ) | \right]^2\le C^2 (b-a)^2  \label{eq:totalorder2},
\end{eqnarray}
and
\begin{eqnarray}
\mathop  {\limsup}_{n \rightarrow \infty} \E \left[ \sup_{0 \le t \le T} \frac{1}{n^{2}}
\sum_{i: \frac{i}{n} \in [a,b]} |\mathcal{X}^n_{i} (nt ) |
\right]\le C (b-a) \label{eq:totalordersmall}.
\end{eqnarray}
In addition for any $\sigma \in [0,T]$ we have
\begin{eqnarray}
\mathop {\lim\sup}_{n \rightarrow \infty}  \E \left[ \sup_{0 \le s \le T - \sigma} \int_s^{s + \sigma} \frac{1}{n^{2}} \sum_{k=1}^{n}
|\mathcal{X}^n_k(n u)| du \right] \le C \sigma.  \label{eq:inttotal}
\end{eqnarray}
\end{lemma}

We now prove Lemma~\ref{lem:boundEdiff}.
\begin{proof}[Proof of Lemma~\ref{lem:boundEdiff}]
We focus on proving \eqref{eq:bounddiffresult}. The proof for \eqref{eq:bounddiffresult2} follows similarly and hence is omitted. Throughout the proof, we use a generic constant $C$ that may depend on $T$ and may vary from line to line, but $C$ is independent of $n$.

From the definitions of $\zeta^{n,+}, \zeta^{n}$ and $f_{\gamma}$ (see \eqref{eq:empirialshapeplus}, \eqref{eq:empirialshape} and \eqref{eq:fgamma}) we derive
\begin{eqnarray} \label{eq:bounddiff}
 \lefteqn{\E \Big[\sup_{t\in [0, T]} \left|  \langle \zeta^{n,+}_{t}, f \rangle -  \langle \zeta^{n}_{t}, f_{\gamma} \rangle \right| \Big]} \nonumber \\
 &=&
 \frac{1}{n^{2}}\E \left[ \sup_{t\in [0, T]} \left|  \sum_{i> p_B^n(nt)} \mathcal{X}^n_i(nt) f(\frac{i}{n}) - \sum_{i=1}^{n} \mathcal{X}^n_i(nt) f_{\gamma}(\frac{i}{n}) \right| \right] \nonumber \\
 &=& \frac{1}{n^{2}} \E \left[ \sup_{t\in [0, T]} \left|  \sum_{i> p_B^n(nt)} \mathcal{X}^n_i(nt) f(\frac{i}{n}) - \sum_{i\ge n \mathsf{p}} \mathcal{X}^n_i(nt) f(\frac{i}{n}) - \sum_{i \in (n\mathsf{p}-n\gamma, n\mathsf{p})} \mathcal{X}^n_i(nt) f_{\gamma}(\frac{i}{n}) \right| \right] \nonumber\\
 &\le& \frac{1}{n^{2}} \E \left[\sup_{t\in [0, T]} \left|  \sum_{i> p_B^n(nt)} \mathcal{X}^n_i(nt) f(\frac{i}{n}) - \sum_{i\ge n \mathsf{p}} \mathcal{X}^n_i(nt) f(\frac{i}{n}) \right| \right] \nonumber \\
 & &+ \frac{1}{n^{2}} \E \left[ \sup_{t\in [0, T]} \left| \sum_{i \in (n\mathsf{p}-n\gamma, n\mathsf{p})} \mathcal{X}^n_i(nt) f_{\gamma}(\frac{i}{n}) \right| \right].
  \end{eqnarray}

 \red{We now bound the two terms after the inequality sign in (\ref{eq:bounddiff}).
We first bound the second term, i.e., the last display in (\ref{eq:bounddiff}).}
Since the family of functions $f_{\gamma}$ are constructed such that $f_{\gamma}$ are uniformly bounded by some constant $C$ on $[0,1]$ for all $\gamma \in (0, \mathsf{p})$, the second term in (\ref{eq:bounddiff}) is bounded above by
\begin{displaymath}
\frac{C}{n^{2}} \cdot \E \left[\sup_{t\in [0, T]} \sum_{i \in (n\mathsf{p}-n\gamma, n\mathsf{p})} |\mathcal{X}^n_i(nt)| \right],
\end{displaymath}
which is further bounded above by $C \gamma$ for large $n$ by (\ref{eq:totalordersmall}).

We next bound the first term, which we denote by $\tilde a_n$. It is clear that $\tilde a_n \ge 0$, so we only have to show that $\lim_{n \rightarrow \infty} \tilde a_n = 0.$ We split it into two parts and study them separately. Fix some $\delta>0$. Since $f$ is continuous, it is bounded by some constant $C$ on $[0,1]$. We then deduce that
\begin{eqnarray}
 \lefteqn{\frac{1}{n^{2}} \E \left[ \sup_{t\in [0, T]} \left|  \sum_{i> p_B^n(nt)} \mathcal{X}^n_i(nt) f(\frac{i}{n}) - \sum_{i\ge n\mathsf{p}} \mathcal{X}^n_i(nt) f(\frac{i}{n}) \right|; \sup_{t \in [0,T]} \left|\frac{p_B^n(nt)}{n}-\mathsf{p} \right| \ge \delta \right] }\nonumber \\
 &\le&
2C \cdot \E \left[ \sup_{0 \le s \le T} \sum_{i=1}^{n} \frac{1}{n^{2}}
\left| \mathcal{X}^n_i(ns) \right|; \sup_{ t \in [0, T]}
\left|\frac{p_B^n(nt)}{n}-\mathsf{p} \right| \ge \delta \right] \buildrel \Delta \over =  b_{n,\delta}\label{eq:bounddiff11}.
\end{eqnarray}
Now (\ref{eq:totalorder2}) in Lemma~\ref{lem:totalorder} implies that the
sequence $\{ \sup_{0 \le s \le T} \sum_{k=1}^{n} \frac{1}{n^{2}}
\left| \mathcal{X}^n_k(ns) \right|: n \ge 1\}$ is uniformly bounded in $L^2$,
thus $\{\sup_{0 \le s \le T} \sum_{k=1}^{n} \frac{1}{n^{2}} \left|
\mathcal{X}^n_k(ns) \right|: n \ge 1\}$ is uniformly integrable. Thus we obtain from part (a) of Theorem~\ref{lem:ode} and \eqref{eq:bounddiff11} that for any $\delta>0$,
\begin{eqnarray} \label{eq:bndelta}
\lim_{n \rightarrow \infty} b_{n,\delta} =0.
\end{eqnarray}
We next proceed to bound
\begin{eqnarray*}
 \frac{1}{n^{2}} \E \left[ \sup_{t\in [0, T]} \left|  \sum_{i> p_B^n(nt)} \mathcal{X}^n_i(nt) f(\frac{i}{n}) - \sum_{i\ge n\mathsf{p}} \mathcal{X}^n_i(nt) f(\frac{i}{n}) \right|; \sup_{t \in [0,T]} \left|\frac{p_B^n(nt)}{n}-\mathsf{p} \right| < \delta \right] \label{eq:bounddiff12}.
\end{eqnarray*}
We can choose small $\delta>0$ such that $ (\mathsf{p}-\delta, \mathsf{p}+\delta) \subset (0,1)$. The above term is then bounded above by
\begin{displaymath}
 \frac{C}{n^{2}}\cdot \E \left[ \sup_{t\in [0, T]} \sum_{i \in (n\mathsf{p}-n\delta, n\mathsf{p} + n \delta)} |\mathcal{X}^n_i(nt)| \right],
\end{displaymath}
which is bounded above by $C \delta$ for $n$ large using (\ref{eq:totalordersmall}).
Upon combining the two parts we find for large $n$
\begin{eqnarray*}
0 \le \tilde a_n \le b_{n, \delta}+  C \delta \quad \text{for sufficiently small $\delta>0$.}
\end{eqnarray*}
Fix $\delta>0$ and let $n \rightarrow \infty$. We infer from \eqref{eq:bndelta} that
\begin{eqnarray*}
0 \le \mathop{\limsup}_{n \rightarrow \infty}\tilde a_n \le   C \delta.
\end{eqnarray*}
Let $\delta \rightarrow 0+$, we obtain
$\lim_{n \rightarrow \infty}\tilde a_n =0$. Therefore the proof is complete.
\end{proof}

\begin{proof} [Proof of Lemma~\ref{lem:totalorder}]
It suffices to prove
(\ref{eq:totalorder2}) since it is clear that (\ref{eq:totalorder2}) implies \eqref{eq:totalordersmall} \red{by the Jensen's inequality. Inequality \eqref{eq:totalordersmall}} further leads to (\ref{eq:inttotal}) by observing that \red{
\begin{eqnarray*}
 \E \left[ \sup_{0 \le s \le T - \sigma} \int_s^{s + \sigma} \frac{1}{n^{2}} \sum_{k=1}^{n}
|\mathcal{X}^n_k(n u)| du \right] &\le &  \E \left[ \sup_{0 \le s \le T - \sigma} \int_s^{s + \sigma} \sup_{u\in [s, s+\sigma]}\frac{1}{n^{2}} \sum_{k=1}^{n}
|\mathcal{X}^n_k(n u)| du \right] \\
&\le & \sigma \cdot \E \left[ \sup_{0 \le s \le T - \sigma} \sup_{u\in [s, s+\sigma]}\frac{1}{n^{2}} \sum_{k=1}^{n}
|\mathcal{X}^n_k(n u)| \right] \\
&= & \sigma \cdot \E \left[\sup_{u \le T} \frac{1}{n^2}\sum_{k=1}^{n}
|\mathcal{X}^n_k(nu)| \right].
\end{eqnarray*}}

Thus, the rest of the proof focus on establishing (\ref{eq:totalorder2}).
For fixed $n$ and $i \in \{1, \ldots, n\}$, we write $A^n_i(t)$ for the total number of new limit orders submitted at price level $i/n$ by time $t$. One can construct independent Poisson processes $\Bar{A}^n_1, \ldots, \Bar{A}^n_n $ with rate $ \Bar{\Lambda}$ given in \eqref{eq:lambda-max} such that for each $i$, $\Bar{A}^n_i$
is defined on the probability space which $A^n_i$ lives in, and
\begin{eqnarray} \label{eq:tildeAstoc}
\Bar{A}^n_i (t, \omega) \ge {A}^n_i (t, \omega) \quad \text{for every sample
path $\omega$ and every $t \ge 0$.}
\end{eqnarray}
Set
\begin{eqnarray*}
\Bar {\mathcal{{E}}}^n(t) \buildrel \Delta \over = \sum_{i: \frac{i}{n} \in [a,b]} \Bar {A}^n_i (nt) \quad \mbox{and} \quad
\mathcal{N}^n(t) = \sum_{i: \frac{i}{n} \in [a,b]} |\mathcal{X}^n_{i} (nt )|.
\end{eqnarray*}
We immediately get that $\Bar{\mathcal{{E}}}^n(\cdot)$ is a Poisson process with rate $r_n$ where
\begin{eqnarray} \label{eq:rn}
r_n \buildrel \Delta \over = n \Bar{\Lambda} \cdot (\lfloor nb \rfloor- \lceil na \rceil +1) \le n \Bar{\Lambda} \cdot ( nb -  na +1),
\end{eqnarray}
and
\[\sup_{0 \le t \le T} {\mathcal{{N}}}^n (t) \le {\mathcal{{N}}}^n(0) + \Bar{\mathcal{{E}}}^n (T) \quad \text{for every sample path.} \]
This implies that
\begin{eqnarray} \label{eq:Ntbound}
\E [\sup_{0 \le t \le T} \mathcal{N}^n (t)]^2 \le \E  [\mathcal{N}^n (0) +  \Bar {\mathcal{{E}}}^n (T)]^2 \le
2 \E [\mathcal{N}^n(0)]^2 + 2 \E [ \Bar{\mathcal{E}}^n(T)]^2.
\end{eqnarray}
One readily checks from Assumption~\ref{ass:initialcondition} that
\begin{eqnarray} \label{eq:totalinitialnsqr}
\frac{1}{n^{4}}\E [\mathcal{N}^n(0)]^2 = 
\Big( \frac{1}{n} \sum_{i: \frac{i}{n} \in [a,b]} |\varrho(i/n)| \Big)^2 \le \left[C {\Bigl((b-a)+ \frac{1}{n}\Bigr)}\right]^2,
\end{eqnarray}
where we use the fact that $\varrho$ is bounded by some constant $C$ on $[0,1]$.
To bound $\E [\Bar {\mathcal{E}}^n(T)]^2$ in (\ref{eq:Ntbound}),
we note that $\Bar{\mathcal{E}}^n$ is a Poisson process with rate $r_n.$ Thus we
have
\begin{eqnarray} \label{eq:tildeEbound}
\E [\Bar {\mathcal{E}}^n(T)]^2 = r_n T + r_n^{2} T^2.
\end{eqnarray}
One readily verifies from \eqref{eq:rn} and \eqref{eq:tildeEbound} that
\begin{eqnarray} \label{eq:tildeEbound2}
\mathop{\lim \sup}_{n \rightarrow \infty} \frac{1}{n^4}\E [\Bar{\mathcal{E}}^n(T)]^2 = {\bar \Lambda^2} \cdot (b-a)^2 {T^2}.
\end{eqnarray}
Now (\ref{eq:totalorder2}) follows from (\ref{eq:Ntbound}),
(\ref{eq:totalinitialnsqr}) and (\ref{eq:tildeEbound2}). The proof is complete.
\end{proof}

\section{Proof of Proposition~\ref{lem:tightposneg0}}\label{sec:appendixA}
It is the goal of this appendix to establish the tightness of $\{(\zeta^{n,+}, \zeta^{n,-}) : n \ge 1\}$,
i.e., to prove Proposition~\ref{lem:tightposneg0}. Our approach is to first show the tightness of $\{ \zeta^{n, +}: n \ge 1 \}$ and $\{ \zeta^{n,-}: n \ge 1 \}$ individually, and then obtain the joint tightness by proving the limit points of $\{ \zeta^{n, +}: n \ge 1 \}$ and $\{ \zeta^{n,-}: n \ge 1 \}$ are concentrated on the set of continuous paths. Throughout the proofs in this section, we use a generic constant $C$ which may vary from line to line, but $C$ is independent of $n$.

We start with several auxiliary lemmas. The next lemma says that in order to establish the tightness of $\{ \zeta^{n, +}: n \ge 1 \}$ and $\{ \zeta^{n,-}: n \ge 1 \}$, it suffices to show \red{that,}
for every $f \in C[0,1]$, the sequences of real-valued processes $\{ \langle \zeta^{n,+}, f \rangle: n \ge 1 \}$ and $\{ \langle \zeta^{n,-}, f \rangle: n \ge 1 \}$ are both
 tight in $\mathbb{D}([0, T]; \mathbb{R})$. This lemma can be found in~\cite[Chapter 4, Proposition 1.7]{kipnis1999scaling}.

\begin{lemma} \label{lem:tightmeasuretoR}
A family of nonnegative measure-valued processes $\{\nu^n_{}: n \ge 1 \}$
is tight in $\mathbb{D}([0, T], \mathcal{M^+}[0,1])$ if $\{ \langle \nu^n_{}, f \rangle: n \ge 1 \}$ is tight in $\mathbb{D}([0, T]; \mathbb{R})$ for every $f \in C[0,1]$.
\end{lemma}

The next lemma is Proposition~VI.3.26 in \cite{Shiryaev03}. 


\begin{lemma} \label{lem:relativecompact0}
 For fixed $T>0$, let
$\{\mathbb{X}_{}^n : n \ge 1\}$ be a sequence of processes that take values in $\mathbb{D}([0,T], \mathbb{R}^d)$ equipped with Skorokhod $J_1$ topology. For each $n \ge 1,$ $\{\mathbb{X}_{}^n\}$ is adapted to a filtration $(\mathcal{F}^n_t)_{t \in [0, T]}$.
 Then $\{\mathbb{X}_{}^n : n \ge 1\}$  is $C$--tight in $\mathbb{D}([0,T], \mathbb{R}^d)$ if and only if the following two conditions hold:
\begin{itemize}
\item [(a)] For any $t \in [0, T]$ and $\epsilon>0,$ there exists a compact set
$\mathsf{K}(t, \epsilon) \subset \mathbb{R}^d$ such that
\begin{eqnarray} \label{eq:relativecomp1}
\inf_{n}\Prob (\mathbb{X}_t^n \in \mathsf{K}(t, \epsilon)) > 1-\epsilon.
\end{eqnarray}
\item [(b)] \red{For every $\eps>0$,}
\begin{equation}
  \lim_{\sigma \to 0}\limsup_{n \to \infty} \Prob\left(\sup_{|s-t|\le \sigma, 0 \le s, t \le T} |\mathbb{X}^n_t  -\mathbb{X}^n_{s}|>\eps\right)=0.
  \label{eq:relativecomp2}
  \end{equation}
\end{itemize}
\end{lemma}
For Lemma~\ref{lem:relativecompact0} and its applications in this paper, we always take filtration $(\mathcal{F}^n_t)_{t \in [0, T]}$ as
the natural filtration generated by $\{\mathcal{X}^n(nt): t \in [0, T] \}$ for fixed $n\ge 1$. That is,
\begin{eqnarray} \label{eq:filtration}
\mathcal{F}^n_t \buildrel \Delta \over =  \sigma(\mathcal{X}^n(ns) , s \le t).
\end{eqnarray}

The next lemma is Corollary~VI.3.33 in \cite{Shiryaev03}.

\begin{lemma} \label{lem:joint-tight}
\red{Let $\{\mathbb{X}_{}^n: n \ge 1\}$ and $\{\mathbb{Y}_{}^n : n \ge 1\} $ be two $C$--tight sequences of processes in $\mathbb{D}([0,T], \mathbb{R}^d)$. Then $\{\mathbb{X}_{}^n + \mathbb{Y}_{}^n : n \ge 1\}$ is $C$--tight in $\mathbb{D}([0,T], \mathbb{R}^d)$ and
$\{(\mathbb{X}_{}^n , \mathbb{Y}_{}^n) : n \ge 1\}$ is $C$--tight in $\mathbb{D}([0,T], \mathbb{R}^{2d})$.}
\end{lemma}

The next lemma states the $C$--tightness of the sequence of real-valued processes $\{\langle \zeta^n, f \rangle: n \ge 1 \}$ for each $f \in {C}([0,1])$. {The proof is lengthy so we defer it to the end of this appendix}. Recall from $(\ref{eq:empirialshape})$ that we have for each $n \ge 1$ and $t \in [0,T]$,
\begin{eqnarray*} 
\langle \zeta_t^n, f \rangle
 =\frac{1}{n^{2}} \sum_{i=1}^{n} \mathcal{X}^n_{i} (nt )
\cdot f ( \frac{i}{n}).
\end{eqnarray*}
\begin{lemma} \label{lem:relativecompact}
Fix any $f \in C([0,1])$. The sequence of real-valued stochastic
processes $\{\langle \zeta^n, f \rangle : n \ge 1\}$ is
$C$--tight in $\mathbb{D}([0,T], \R)$.
\end{lemma}

With Lemmas~\ref{lem:tightmeasuretoR}--\ref{lem:relativecompact} at our disposal, we are ready to prove Proposition~\ref{lem:tightposneg0}.

\begin{proof}[Proof of Proposition~\ref{lem:tightposneg0}]
 We apply Lemma~\ref{lem:relativecompact0} to establish the $C$--tightness of $\{\langle \zeta^{n,+}, f \rangle: n \ge 1\}$ in $\mathbb{D}([0,T], \mathbb{R})$ for fixed $f \in C([0,1])$, where
 \begin{eqnarray} \label{eq:zetatnfPlus}
\langle \zeta_t^{n, +}, f \rangle
 =\frac{1}{n^{2}} \sum_{i=1}^{n} \Big[\mathcal{X}^n_{i} (nt )\Big]^+
\cdot f ( \frac{i}{n}).
\end{eqnarray}
 The $C$--tightness of $\{\langle \zeta^{n,-}, f \rangle: n \ge 1\}$ follows using a similar argument. The $C$--tightness of $\left\{ \left(\langle \zeta^{n,+}, f \rangle , \langle \zeta^{n,-}, f \rangle \right) : n \ge 1 \right\}$ immediately follows after applying Lemma~\ref{lem:joint-tight}.

We first show that $\{\langle \zeta^{n,+}, f \rangle: n \ge 1\}$ satisfies part (a) of Lemma~\ref{lem:relativecompact0}. We rely on the family of functions $f_{\gamma} \in C[0,1]$ introduced in \eqref{eq:fgamma}.
Note that for fixed $t \in [0, T],$
\begin{eqnarray*} \label{eq:plusrelativecomp1}
\lefteqn{\sup_{n}\Prob (|\langle \zeta^{n,+}_{t}, f \rangle| > L) } \\&\le&
\sup_{n}\Prob (|\langle \zeta^{n,+}_{t}, f \rangle - \langle \zeta^{n}_{t}, f_{\gamma} \rangle | > \frac{L}{2}) + \sup_{n}\Prob (|\langle \zeta^{n}_{t}, f_{\gamma} \rangle| > \frac{L}{2}) \\
&\le& \frac{2}{L}\sup_{n} \left( \tilde a_n +  C \gamma  \right)+ \sup_{n}\Prob (|\langle \zeta^{n}_{t}, f_{\gamma} \rangle| > \frac{L}{2}),
\end{eqnarray*}
where we use Markov inequality and (\ref{eq:bounddiffresult}) in the last inequality.
Since $\{ \tilde a_n \}$ is a bounded sequence, we can choose $L$ large such that $\frac{2}{L}\sup_{n} \left( \tilde a_n +  C \gamma  \right)$ is arbitrarily small.
 In addition, for fixed $\gamma$ and $t$, $\{\langle \zeta^{n}_{t}, f_{\gamma} \rangle: n \ge 1\}$ is a tight sequence by Lemma~\ref{lem:relativecompact}. Thus we can pick $L$ large such that $\sup_{n}\Prob (|\langle \zeta^{n,+}_{t}, f \rangle| > L)$ is also arbitrarily small. These two facts imply that $\{\langle \zeta^{n,+}_{t}, f \rangle: n \ge 1\}$ is a tight sequence for fixed $t \in [0,T]$.

{We next verify that $\{\langle \zeta^{n,+}, f \rangle: n \ge 1\}$ also satisfies part (b) of Lemma~\ref{lem:relativecompact0}. Clearly from \eqref{eq:zetatnfPlus} we have for fixed $n \ge 1$, the process $\langle \zeta^{n,+}, f \rangle$ is adapted to the filtration $\mathcal{F}^n$ given in \eqref{eq:filtration}.
Given a real number $c>0$, we obtain from
(\ref{eq:bounddiffresult}) that for $n$ large
\begin{eqnarray}\label{eq:b}
\lefteqn{\Prob \left(\sup_{|s-t|\le \sigma, 0 \le s, t \le T} |\langle \zeta^{n,+}_{t}, f \rangle-  \langle \zeta^{n,+}_{s}, f \rangle | > c \right)} \nonumber \\
&\le & 2 \cdot \Prob\left(\sup_{0 \le t \le T} |\langle \zeta^{n,+}_{t}, f \rangle-  \langle \zeta^{n}_{t}, f_{\gamma} \rangle |> \frac{c}{3}  \right)  +  \Prob \left( \sup_{|s-t|\le \sigma, 0 \le s, t \le T} |\langle \zeta^{n}_{t}, f_{\gamma} \rangle-  \langle \zeta^{n}_{s}, f_{\gamma} \rangle |> \frac{c}{3} \right) \nonumber \\
&\le& 2 \cdot \frac{3}{c} \left( \tilde a_n + C \gamma\right) + \Prob\left(\sup_{|s-t|\le \sigma, 0 \le s, t \le T} |\langle \zeta^{n}_{t}, f_{\gamma} \rangle-  \langle \zeta^{n}_{s}, f_{\gamma} \rangle |> \frac{c}{3} \right).
\end{eqnarray}
Since for fixed $\gamma>0$, the process $\langle \zeta^{n}, f_{\gamma} \rangle$ is adapted to $\mathcal{F}^n$,
we infer from Lemma~\ref{lem:relativecompact} that
 \begin{eqnarray*}
 \lim_{\sigma \to 0}\limsup_{n \to \infty} \Prob\left( \sup_{|s-t|\le \sigma, 0 \le s, t \le T} |\langle \zeta^{n}_{t}, f_{\gamma} \rangle-  \langle \zeta^{n}_{s}, f_{\gamma} \rangle |> \frac{c}{3} \right) =0.
\end{eqnarray*}
In conjunction with \eqref{eq:b}, this implies for fixed $\gamma \in (0, \mathsf{p})$ and $c>0$,
 \begin{eqnarray*}
 \lim_{\sigma \to 0}\limsup_{n \to \infty} \Prob \left( \sup_{|s-t|\le \sigma, 0 \le s, t \le T} |\langle \zeta^{n,+}_{t}, f \rangle-  \langle \zeta^{n,+}_{s}, f \rangle | > c \right) \le \frac{6C\gamma}{c}.
\end{eqnarray*}
Let $\gamma \rightarrow 0+,$ we obtain
 \begin{eqnarray*}
 \lim_{\sigma \to 0}\limsup_{n \to \infty} \Prob \left( \sup_{|s-t|\le \sigma, 0 \le s, t \le T} |\langle \zeta^{n,+}_{t}, f \rangle-  \langle \zeta^{n,+}_{s}, f \rangle | > c \right)=0.
\end{eqnarray*}
}

\red{Given the $C$--tightness of $\{\langle \zeta^{n,+}, f \rangle: n \ge 1\}$ and $\{\langle \zeta^{n,-}, f \rangle: n \ge 1\}$ for any fixed function $f \in C[0,1]$, we first deduce from Lemma~\ref{lem:tightmeasuretoR}  that both $\{\zeta^{n,+}: n \ge 1\}$ and $\{\zeta^{n,-}: n \ge 1\}$ are tight in $\mathbb{D}([0, T], \mathcal{M^+}[0,1])$. We next argue that the limit points of $\{\zeta^{n,+}: n \ge 1\}$ and $\{\zeta^{n,-}: n \ge 1\}$ have continuous sample paths. Suppose $\zeta^+$ is a limit point of $\{\zeta^{n,+}: n \ge 1\}$, then we readily verify from continuous mapping theorem that $\langle \zeta^+, f \rangle $ is a limit point of $\{ \langle \zeta^{n,+} f \rangle : n \ge 1\}$ for any function $f \in C[0,1]$. Since all the limit points of the $C$--tight sequence $\{\langle \zeta^{n,+}, f \rangle: n \ge 1\}$ are concentrated on the set of continuous paths, we deduce that for almost every $\omega$ we have: for fixed $t \in [0, T]$,
\begin{eqnarray}
\lim_{s \rightarrow t}  d_+\left(\zeta^+_{s}(\omega) , \zeta^+_{t}(\omega) \right) = \lim_{s \rightarrow t} \sum_{k=1}^\infty \frac{1}{2^k} \frac{|\langle \zeta^+_{s}(\omega), \phi_k \rangle -\langle \zeta^+_{t}(\omega), \phi_k \rangle|}
{1+|\langle \zeta^+_{s}(\omega), \phi_k \rangle -\langle \zeta^+_{t}(\omega), \phi_k \rangle| }   =0,
\end{eqnarray}
where $(\phi_k)$ is a dense subset of $C[0,1]$ and the distant measure $d_+$ is given in \eqref{eq:metricplus}. Thus the limit point $\zeta^+$ has continuous sample paths. Similarly, almost every path of each limit point of $\{\zeta^{n,-}: n \ge 1\}$ is continuous.
Now the tightness of $\{(\zeta^{n,+}, \zeta^{n,-}) : n \ge 1\}$ in $\D([0, T], {\mathcal{\overline {M}}}([0,1]) )$ follows from Lemma~3.5 and Corollary~3.6 in~\cite{aldous1981weak}.}
\end{proof}



\begin{proof}[Proof of Lemma~\ref{lem:relativecompact}]
{We verify that the conditions (a) and (b) in  Lemma~\ref{lem:relativecompact0} are satisfied.}
For notational conveniences, we write for
a fixed function $f \in C[0,1]$ and for $t \ge 0$
\begin{eqnarray} \label{eq:Ypitnf}
\mathcal{Y}_t^n \buildrel \Delta \over = \langle \zeta_t^n, f \rangle &
 =& \frac{1}{n^{2}} \sum_{i=1}^{n} \mathcal{X}^n_{i} (nt )
\cdot f ( \frac{i}{n}).
\end{eqnarray}

We first prove that $\mathcal{Y}^n$ satisfies condition (a) from Lemma~\ref{lem:relativecompact0}. Since $\sup_{x \in [0,1]} |f(x)| \le C$ for some constant $C$ by the continuity of $f$, we deduce from
(\ref{eq:Ypitnf}) that
\begin{eqnarray*}\label{eq:boundYt}
|\mathcal{Y}_t^n | \le \frac{1}{n^{2}} \sum_{i=1}^{n} |\mathcal{X}^n_{i} (nt )| \cdot |f (
\frac{i}{n})| \le \frac{C}{n^{2}} \sum_{i=1}^{n} |\mathcal{X}^n_{i} (nt )|.
\end{eqnarray*}
We now deduce from Lemma~\ref{lem:totalorder} that
\begin{eqnarray} \label{eq:boundytn}
\sup_{n}\E \Big[ \sup_{0 \le t \le T}|\mathcal{Y}_t^n |^2\Big] \le C.
\end{eqnarray}
An application of Markov's inequality immediately yields that $\mathcal{Y}_t^n$ satisfies
(\ref{eq:relativecomp1}).

Next we show that $\mathcal{Y}^n$ also satisfies condition (b) in Lemma~\ref{lem:relativecompact0}. Note that
\begin{eqnarray}
\mathcal{Y}_t^n = F( \mathcal{X}^n (nt) ) =\mathcal{Y}_0^n +\int_0^{nt}
\mathcal{L}_nF(\mathcal{X}^n (s))ds+ \mathsf{M}^n_t, \label{eq:Yn}
\end{eqnarray}
where \red{$\mathcal{L}_n$ is the operator given in (\ref{eq:generator})}, $\mathsf{M}^n$ is a (local) martingale, and the function $F$ is defined in \eqref{eq:F}.
Given $\epsilon>0, \sigma>0$, we deduce
from (\ref{eq:Yn}) that
\begin{eqnarray} \label{eq:twoparts}
 \lefteqn{\Prob  \left( \sup_{|s-t|\le \sigma, 0 \le s, t \le T}  |\mathcal{Y}^n_{t}-\mathcal{Y}^n_{s}|>\eps \right)} \nonumber \\
 &\leq & \Pb \left( \sup_{0 \le t- s \le \sigma, 0 \le s, t \le T} \left|\int_{ns}^{nt} \mathcal{L}_n F(\mathcal{X}^n(u))du\right|>\frac{\eps}{2} \right) \nonumber \\
 &+&\Pb \left( \sup_{|s-t|\le \sigma, 0 \le s, t \le T}  |\mathsf{M}^n_{t}-\mathsf{M}^n_s|>\frac{\eps}{2}\right).
\end{eqnarray}
Our strategy is to show for large $n$, there is a constant $C$ independent of $n$
such that
\begin{eqnarray}
\E \left[ \sup_{0 \le t- s \le \sigma, 0 \le s, t \le T} \left|\int_{ns}^{nt} \mathcal{L}_n F(\mathcal{X}^n(u))du\right| \right]
&\leq&  C \sigma,\label{eq:boundEgenrator} \\
\E \left[ \sup_{|s-t|\le \sigma, 0 \le s, t \le T}  |\mathsf{M}^n_{t}-\mathsf{M}^n_s|^2\right] &\leq & \frac{C T}{n^{2}}. \label{eq:Mn2ndmmtbound}
\end{eqnarray}
Using Markov's inequality and Chebyshev's
inequality, and (\ref{eq:twoparts}), we obtain
\begin{equation*}\label{eq:relativecompY}
  \lim_{\sigma \to 0}\limsup_{n \to \infty} \Prob  \left( \sup_{|s-t|\le \sigma, 0 \le s, t \le T}  |\mathcal{Y}^n_{t}-\mathcal{Y}^n_{s}|>\eps \right) =0.
\end{equation*}

The rest of the proof focuses on establishing
(\ref{eq:boundEgenrator}) and (\ref{eq:Mn2ndmmtbound}). We start
with proving (\ref{eq:boundEgenrator}). For fixed $n$, we deduce from (\ref{eq:F}) and (\ref{eq:generator}) that
\begin{equation} \label{eq:Fgenerator}
\begin{aligned}
\mathcal{L}_nF(\mathcal{X}^n(u))
=&\sum_{k< p^n_A(u)}\left[\frac{1}{n^{2}}f(\frac{k}{n})\Theta_B^n({p^n_A(u)-k})|\mathcal{X}^n_k(u)|-\frac{1}{n^{2}}f(\frac{k}{n})\Lambda_B^n({p^n_A(u)-k})\right]\\
+&\sum_{k> p^n_B(u)}\left[\frac{1}{n^{2}}f(\frac{k}{n})\Lambda_A^n({k-p^n_B(u)})-\frac{1}{n^{2}}f(\frac{k}{n})\Theta_A^n({k-p^n_B(u)})|\mathcal{X}^n_k(u)|\right]\\
+&\frac{1}{n^{2}}\left(f(\frac{p^n_B(u)}{n}) \Upsilon_A^n -f(\frac{p^n_A(u)}{n}) \Upsilon_B^n
\right).
\end{aligned}
\end{equation}
Since $f$ is bounded on $[0,1]$, we deduce from
Assumption~\ref{ass:rateconverg1} that
\begin{eqnarray}
{|\mathcal{L}_nF(\mathcal{X}^n (u))|}
 &\le&  \frac{C}{ n^{3}} \sum_{k=1}^n |\mathcal{X}^n_k(u)| + \frac{ C}{n}. \label{eq:boundLnF}
\end{eqnarray}
This immediately yields that \red{
\begin{eqnarray*}
\E \left[ \sup_{0 \le t- s \le \sigma, 0 \le s, t \le T}\int_{ns}^{nt}\left| \mathcal{L}_n F(\mathcal{X}^n(u))\right| du \right]
\leq  \frac{C}{n^{2}} \E \left[ \sup_{0 \le s \le T - \sigma} \int_s^{s + \sigma} \sum_{k=1}^{n}
|\mathcal{X}^n_k(n u)| du \right]+ C \sigma.
\end{eqnarray*}
On combining (\ref{eq:inttotal}), we obtain (\ref{eq:boundEgenrator}).}

We next show (\ref{eq:Mn2ndmmtbound}). To this end, we define for each $u \ge 0,$
\begin{equation*}\label{eq:etant}
\eta^n (u)=\mathcal{L}_nF^2(\mathcal{X}^n(u))-2F(\mathcal{X}^n(u)) \cdot \mathcal{L}_n F(\mathcal{X}^n(u)),
\end{equation*}
where $\mathcal{L}_n$ is given in (\ref{eq:generator}) and $F$ is the linear function given in (\ref{eq:F}). One checks that
\begin{equation*}
\begin{aligned}
\eta^n(u) =&\sum_{k< p^n_A(u)}\left[\frac{1}{n^{4}}f(\frac{k}{n})^2 \cdot \Theta_B^n({p^n_A(u)-k}) \cdot |\mathcal{X}^n_k(u)|+\frac{1}{n^{4}}f(\frac{k}{n})^2 \cdot \Lambda_B^n({p^n_A(u)-k})\right]\\
+&\sum_{k>
p^n_B(u)}\left[\frac{1}{n^{4}}f(\frac{k}{n})^2 \cdot \Lambda_A^n({k-p^n_B(u)})+\frac{1}{n^{4}}f(\frac{k}{n})^2 \cdot \Theta_A^n({k-p^n_B(u)}) \cdot |\mathcal{X}^n_k(u)|\right]\\
+&\frac{1}{n^{4}}\left(f(\frac{p^n_A(u)}{n})^2 \cdot \Upsilon_B^n +f(\frac{p^n_B(u)}{n})^2 \cdot \Upsilon_A^n
\right).\label{eq:infigene1}
\end{aligned}
\end{equation*}
Now for each fixed $n,$ one verifies that $\{(\mathsf{M}^n_t)^2-\int_0^{nt}\eta^n(u)du: t \ge
0 \}$ is a local martingale with respect to the filtration $\mathcal{F}^n$ in \eqref{eq:filtration} (see Lemma~5.1 in Appendix 1 of~\cite{kipnis1999scaling}). \red{Suppose $\{(\mathsf{M}^n_t)^2-\int_0^{nt}\eta^n(u)du: t \ge
0 \}$ and $\{\mathsf{M}^n_t: t \ge
0 \}$ are indeed $\mathcal{F}^n-$martingales. Then Doob's maximal inequality for martingales yields that
\begin{eqnarray}\label{eq:mating-equlity}
\E \left[ \sup_{|s-t|\le \sigma, 0 \le s, t \le T}  |\mathsf{M}^n_{t}-\mathsf{M}^n_s|^2\right] \leq
4 \E \left[ \sup_{0 \le  t \le T}  |\mathsf{M}^n_{t}|^2\right]
\le  16 \E (|\mathsf{M}^n_{T}|^2)
= \E\Big[\int_{0}^{nT} \eta^n (u) du\Big] \nonumber\\
\end{eqnarray} }
Hence to establish (\ref{eq:Mn2ndmmtbound}), we proceed to bound $\eta^n(u)$. Using Assumption~\ref{ass:rateconverg1} and the
fact that $f$ is bounded on \red{$[0,1]$,  there is a constant} $C>0$ such that
\begin{eqnarray} \label{eq:boundetanu}
0 \le \eta^n (u) \le \frac{1}{n^{4}} \left[ \frac{C}{n} \sum_{k=1}^{n} |\mathcal{X}^n_k(u) |
+ Cn  \right].
\end{eqnarray}
Upon combining (\ref{eq:inttotal}) and applying the change of variable formula we find for $n$ large
\begin{eqnarray} \label{eq:boundetanu}
\E\Big[\int_{0}^{nT} \eta^n (u) du\Big]  \le \frac{CT}{n^2}.
\end{eqnarray}
We then obtain (\ref{eq:Mn2ndmmtbound}) from (\ref{eq:mating-equlity}).

It remains to show that the two local martingales $\{(\mathsf{M}^n_t)^2-\int_0^{nt}\eta^n(u)du: t \ge
0 \}$ and $\{\mathsf{M}^n_t: t \ge
0 \}$ are indeed $\mathcal{F}^n$-martingale for each fixed $n$. \red{By \cite[Theorem~51]{protter2004stochastic}},
it suffices to show for every $t \le T$,
\begin{eqnarray}
\E \Big[\sup_{ 0 \le  s \le t} |\mathsf{M}_s^n|^2 \Big] < \infty, \label{eq:martingalelocal1} \\
 \E \Big[\int_0^{nt}\eta^n(u)du \Big] < \infty . \label{eq:martingalelocal2}
\end{eqnarray}
Inequality \eqref{eq:martingalelocal2} directly follows from \eqref{eq:boundetanu}. To prove \eqref{eq:martingalelocal1}, we use (\ref{eq:Yn}), which implies
\begin{eqnarray*}
\sup_{ 0 \le  s \le t} |\mathsf{M}_s^n|^2 \le 3 \sup_{ 0 \le  s \le t} (\mathcal{Y}_s^n)^2 + 3 (\mathcal{Y}_0^n)^2 + 3 \Big(\int_0^{nt}
\Big| \mathcal{L}_nF(\mathcal{X}^n (s))\Big| ds  \Big)^2
 \end{eqnarray*}
Inequality (\ref{eq:martingalelocal1}) then follows from (\ref{eq:boundytn}), (\ref{eq:boundLnF}) and (\ref{eq:totalorder2}).
The proof is thus complete.
\end{proof}

\section{Proof of Proposition~\ref{lem:limitpointpi}} \label{app:2}
In this section, we prove Proposition~\ref{lem:limitpointpi}.
We rely on the representation in \eqref{eq:repzetatnf}, i.e.,
\begin{eqnarray} \label{eq:pitngenerator}
\langle \zeta_t^n, f \rangle
& =&\frac{1}{n^{2}} \sum_{i=1}^{n} \mathcal{X}^n_{i} (nt )
\cdot f \Big( \frac{i}{n}\Big) = F( \mathcal{X}^n (nt) ) \nonumber\\
& =& \langle \zeta_0^n, f \rangle + n \int_0^{t}
\mathcal{L}_n F(\mathcal{X}^n(ns))ds+ \mathsf{M}^n_t .
\end{eqnarray}
Here $\mathcal{L}_n$ is operator given in (\ref{eq:generator}), $\mathsf{M}^n$ is a martingale and $F$ is given in \eqref{eq:F}.
We fix $f \in C([0,1])$ throughout this section.
It is clear from \eqref{eq:X0} that
\begin{eqnarray} \label{eq:intialX}
\lim_{n \rightarrow \infty} \langle \zeta_0^n, f \rangle =\int_{0}^1 f(x) \varrho(x) dx= \langle \zeta_0, f \rangle,
\end{eqnarray}
where $\zeta_0$ is a deterministic signed measure with density $\zeta_0(dx)=\varrho(x)dx$ for $x\in [0,1]$.

To prove Proposition~\ref{lem:limitpointpi}, we now introduce two auxiliary lemmas. The next lemma implies that the martingale term in
(\ref{eq:pitngenerator}) vanishes (converges weakly to zero) as $n \rightarrow \infty$. The proof directly follows from \eqref{eq:mating-equlity} and \eqref{eq:boundetanu}.


\begin{lemma} \label{lem:mart0}
For $T>0,$ we have
\begin{eqnarray}\label{eq:mart0}
\lim_{n \rightarrow \infty } \E \left[\sup_{0 \le t \le T} |\mathsf{M}_t^n|^2 \right] = 0.
\end{eqnarray}
\end{lemma}

The next lemma concerns the weak convergence of the ``scaled" generator. The proof is lengthy and thus deferred to the end of this appendix.
\begin{lemma} \label{lem:generatorconvg}
For the subsequence $\{{n_k}: k=1, 2 \ldots\}$ in (\ref{eq:pinconvgpi}) we have for fixed $f \in C([0,1])$
\begin{eqnarray} \label{eq:generatorconvg}
\sup_{0 \le s \le T} \Big| n_k \cdot \mathcal{L}_{n_k} F( \mathcal{X}^{n_k}(n_k s)) - \langle \nu, f
\rangle + \langle \zeta_s, \mathcal{A}_{\Theta}f \rangle \Big|
\Rightarrow 0 \quad \text{as $n_k \rightarrow \infty$},
\end{eqnarray}
where $\nu$ is given in (\ref{eq:nuLambda}),
$\mathcal{A}_{\Theta}$ is given in (\ref{eq:ATheta}) and $F$ is given in (\ref{eq:F}).
\end{lemma}

We are now ready to prove Proposition~\ref{lem:limitpointpi}.

\begin{proof}[Proof of Proposition~\ref{lem:limitpointpi}]
\red{Given \eqref{eq:pitngenerator}--\eqref{eq:generatorconvg} and that $\langle \zeta^{n_k}, f \rangle \Rightarrow \langle \zeta, f \rangle$ {as $n_k \rightarrow \infty$}, we deduce that the process $\zeta$ satisfies Equation (\ref{eq:generatormart3}). The continuity of the path of $\langle \zeta, f \rangle$ directly follows from Lemma~\ref{lem:relativecompact}}.

Next we show that there is an unique measure-valued process satisfying (\ref{eq:generatormart3}). It is evident that $\zeta$ is deterministic since there is no source of randomness in Equation (\ref{eq:generatormart3}). Suppose for
fixed $t\in[0,T]$, $\zeta_t$ is absolutely continuous with respect to Lebesgue
measure and it has a bounded density function $\varphi$ such that $\zeta_t(dx) = \varphi(x,t) dx$. 
After multiplying a smooth test function $g(t)$ on both sides of Equation (\ref{eq:generatormart3}), integrating with respect to $t$ on $[0, T]$, substituting \eqref{eq:nuLambda} and \eqref{eq:ATheta}, and noting that $f$ can be an arbitrary continuous function on $[0,1]$, we readily verify that for almost every $x \in (\mathsf{p},1]$,
\begin{eqnarray*}
\int_{0}^{T} \varphi(x,t) g(t) dt &=& \varrho (x) \int_{0}^{T} g(t) dt \\
&&+ \int_{0}^{T} g(t) \int_{0}^{t} \Big[\Lambda_A ( x - \mathsf{p})  - \Theta_A( x - \mathsf{p}) \varphi (x,s) \Big] ds dt. \nonumber
\end{eqnarray*}
That is, for fixed $x \in (\mathsf{p},1]$, $\varphi(x, \cdot)$ is a weak solution for the ODE given in \eqref{eq:varphiinitial} and \eqref{eq:varphiODE}, \red{and thus $\varphi(x, t)$ is measurable in the second argument~$t$.}
On the other hand,
it is clear that there is an unique classical solution to the ODE \eqref{eq:varphiinitial} and \eqref{eq:varphiODE} for fixed $x \in (\mathsf{p},1]$.
Therefore, we deduce from the equivalence of classical solution and
weak solution for ODE (see, e.g. \cite[Chapter~1, Lemma~1.3]{tao2006nonlinear}) that the density function $\varphi(x, \cdot)$ is the unique classical solution of
(\ref{eq:varphiinitial}) and (\ref{eq:varphiODE}) when $x> \mathsf{p}$. A similar argument yields that $\varphi(x, \cdot)$ is the unique classical solution for
(\ref{eq:varphiinitial}) and (\ref{eq:varphiODE2}) when $x< \mathsf{p}$. As a consequence, the solution for (\ref{eq:generatormart3}) is unique \red{(we allow ourselves the ability to modify $\varphi(x, t)$ on a set of measure zero including the point $x=\mathsf{p}$. The modification is independent of $t$).}

The rest of the proof focuses on showing that $\zeta_{t}$ is absolutely continuous and that it has a bounded density function.
We prove that there exists some constant $C_T$ depending on $T$ such that for all $f \in C([0,1])$ and $t \in [0,T]$,
\[ \left|\langle \zeta_{t}, f \rangle \right| \le C_T \int_{0}^{1} |f(x)|dx.\]
To this end, we first deduce from Equation (\ref{eq:convgzetaf}) and the continuous mapping theorem that for any $f \in C([0,1])$ and $t \in [0,T]$,
\begin{eqnarray*} \label{eq:piabosultevalconvgindistr}
\left|\langle \zeta_{t}^{n_k}, f \rangle \right| \Rightarrow \left|\langle \zeta_{t}, f \rangle \right| \quad \text{as $n_k \rightarrow \infty$}.
\end{eqnarray*}
Applying Theorem~3.4 in \cite{bil99} yields
\begin{eqnarray} \label{eq:fatou}
\E \left|\langle \zeta_{t}, f \rangle \right| \le {\lim \inf}_{n_k \rightarrow \infty} \E \left|\langle \zeta_{t}^{n_k}, f \rangle \right|.
\end{eqnarray}
We next focus on bounding $\E \left|\langle \zeta_{t}^{n_k}, f \rangle \right|$.
Recall the limit order arrival process $A^n_i$ is pathwisely upper bounded by a Poisson process $\bar{A}^n_i$ with rate $ \Bar{\Lambda}$ (see \eqref{eq:tildeAstoc}).
Hence,
\begin{eqnarray*} \label{eq:pitnfexpand}
\left|\langle \zeta_{t}^{n}, f \rangle \right| &=& \left| \frac{1}{n^{2}} \sum_{i=1}^{n} \mathcal{X}^n_{i} (nt )\cdot f ( \frac{i}{n}) \right| \nonumber \\
&\le&  \frac{1}{n^{2}} \sum_{i=1}^{n} |\mathcal{X}^n_{i} (0 )| \cdot |f ( \frac{i}{n})| + \left| \frac{1}{n^{2}} \sum_{i=1}^{n} A^n_{i} (nt )\cdot f ( \frac{i}{n}) \right| \nonumber\\
&\le &  \frac{1}{n} \sum_{i=1}^{n} |\varrho (\frac{i}{n})|\cdot |f ( \frac{i}{n}) | + \frac{1}{n^{2}} \sum_{i=1}^{n} \bar{A}^n_{i} (nt )\cdot \left|f ( \frac{i}{n}) \right|,
\end{eqnarray*}
which implies that
\begin{eqnarray*} \label{eq:abspitnfUI}
\E \left|\langle \zeta_{t}^{n_k}, f \rangle \right| \le \frac{1}{n_k} \sum_{i=1}^{n_k} |\varrho (\frac{i}{n_k})|\cdot | f ( \frac{i}{n_k})|  + \Bar{\Lambda} t \cdot \frac{1}{n_k} \sum_{i=1}^{n_k} |f ( \frac{i}{n_k})|.
\end{eqnarray*}
In conjunction with (\ref{eq:fatou}) and the fact that $\zeta_{t}$ is a deterministic measure, we deduce
\begin{eqnarray} \label{eq:boundzetaf1}
\left|\langle \zeta_{t}^{}, f \rangle \right| \le \int_{0}^1 |\varrho (x)| \cdot | f (x)| dx  + \Bar{\Lambda} t \cdot \int_{0}^1 | f (x)| dx,
\end{eqnarray}
after letting $n_k \rightarrow \infty$. Thus we have all $f \in C([0,1])$ and $t \in [0,T]$,
\begin{eqnarray} \label{eq:boundzetaf3}
\left|\langle \zeta_{t}, f \rangle \right|
&\le& C_T \int_{0}^1 | f (x)| dx,
\end{eqnarray}
where $C_T =   \max_{x \in[0,1]} |\varrho (x)| +  \Bar{\Lambda} T$. \red{Since continuous functions are dense in the space of Lebesgue-integrable functions $L^1([0,1])$, we deduce from the bounded linear transformation theorem that Equation \eqref{eq:boundzetaf3} holds for all $f \in L^1([0,1])$. We then conclude from \cite[Theorem~15.6]{driver2003analysis} that $\zeta_t$ is a finite measure, absolutely continuous with respect to Lebesgue measure, and its density function $\varphi(x,t)$ is bounded by $C_T$ uniformly with respect to $x\in[0,1]$ and $t\in[0,T]$.}
The proof is complete.
\end{proof}


\begin{proof}[Proof of Lemma~\ref{lem:generatorconvg}]
Throughout the proof, we fix $f \in C([0,1])$ and we use $\{n: n \ge 1\}$ instead of its subsequence $\{n_k: k \ge 1\}$ for notational simplicity. We also use a generic constant $C$ which may vary from line to line but $C$ is independent of $n.$

 We first show for $f \in C[0,1]$,
\begin{equation}\label{eq:generatorconvg1}
\sup_{0 \le s \le T} \left|
n \cdot \sum_{k>p_B^n(ns)}\left[\frac{1}{n^2}f(\frac{k}{n})\cdot \Lambda_A^n({k-p_B^n(ns)})\right]
- \int_{\mathsf{p}}^{1} f(x) \Lambda_A(x-\mathsf{p}) dx \right| \Rightarrow 0 \quad
\text{as $n \rightarrow \infty$}.
\end{equation}
To this end, we define for $z \in [0,1]$ and $n \ge 1$
\begin{eqnarray}
h_n(z)&=&
\sum_{\frac{k}{n}>z}\left[\frac{1}{n}f\left(\frac{k}{n}\right)\Lambda_A\left({\frac{k}{n}-z}\right)\right], \label{eq:hn}\\
h(z)&=& \int_{z}^{1} f(x) \Lambda_A(x-z) dx. \label{eq:h}
\end{eqnarray}
 Applying Assumption~\ref{ass:rateconverg1}, we find that (\ref{eq:generatorconvg1}) is equivalent to
\begin{eqnarray} \label{eq:generatorconvg1equiv}
\sup_{0 \le s \le T} \left|
h_n \left(\frac{p_B^n(ns)}{n} \right) - h(\mathsf{p})\right|
\Rightarrow 0, \quad \text{as $n \rightarrow \infty$}.
\end{eqnarray}
Part (a) of Theorem~\ref{lem:ode} implies that when $n\rightarrow \infty$, we have
\begin{eqnarray*}
\frac{p_B^n(n \cdot)}{n} \Rightarrow \mathsf{p} \quad \text{in $\D([0,T], \R)$.}
\end{eqnarray*}
In addition, since $f$ and $\Lambda_A$ are uniformly continuous and bounded on $[0,1]$, it is straightforward to verify that the function $h$ is continuous, and the sequence $\{h_n: n \ge 1\}$ converges uniformly to $h.$ In particular, we obtain for any sequence of real numbers $\{z_n: n \ge 1\}$,
\begin{eqnarray}\label{eq:hnconvg}
\lim_{n \rightarrow \infty} h_n(z_n) = h(\mathsf{p}) \quad \text{if $\lim_{n
\rightarrow \infty} z_n = \mathsf{p}\in (0,1)$}.
\end{eqnarray}
Then generalized continuous mapping theorem~\cite[Theorem~3.4.4]{Whitt02} yields that as $n\rightarrow \infty$
\begin{eqnarray*}
h_n\left(\frac{p_B^n(n \cdot)}{n} \right) \Rightarrow h(\mathsf{p})  \quad \text{in $\D([0,T], \R)$.}
\end{eqnarray*}
Since the Skorohod $J_1$ topology relativized to ${C}([0,T], \R)$ of continuous functions is equivalent to the uniform topology~\cite[Section~12]{bil99},
we obtain (\ref{eq:generatorconvg1equiv}) and \eqref{eq:generatorconvg1}.

Applying a similar argument, we find
\begin{equation}\label{eq:generatorconvg2}
\sup_{0 \le s \le T}
\left|\sum_{k<p_A^n(ns)}\left[\frac{1}{n}f(\frac{k}{n})\cdot \Lambda_B^n({p_A^n(ns)-k})\right]
- \int_{0}^{\mathsf{p}} f(x) \Lambda_B(\mathsf{p}-x) dx \right| \Rightarrow 0 \quad
\text{as $n \rightarrow \infty$}.
\end{equation}

We next prove that when $n \rightarrow \infty$,
\begin{equation} \label{eq:generatorconvg3}
\sup_{0 \le s \le
T}\left|\sum_{k>p_B^n(ns)}\left[\frac{1}{n}f\left(\frac{k}{n}\right)\cdot\Theta_A^n({k-p_B^n(ns)})
\cdot \mathcal{X}^n_k(ns) \right] - \int_{\mathsf{p}}^{1} f(x) \Theta_A(x-\mathsf{p}) d \zeta_s(x)
\right| \Rightarrow 0.
\end{equation}
It suffices to show
\begin{eqnarray} \label{eq:generatorconvg3a}
\sup_{0 \le s \le
T}\left| \int_{\mathsf{p}}^{1} f(x) \Theta_A(x-\mathsf{p}) d \zeta^{n}_s(x) - \int_{\mathsf{p}}^{1} f(x) \Theta_A(x-\mathsf{p}) d \zeta_s(x)
\right| \Rightarrow 0, \label{eq:generatorconvg3a}
\end{eqnarray}
and
\begin{equation} \label{eq:generatorconvg3c}
\sup_{0 \le s \le
T}\left|\sum_{k>p_B^n(ns)}\left[\frac{1}{n}f\left(\frac{k}{n}\right)\Theta_A^n({k-p_B^n(ns)})
\cdot \mathcal{X}^n_k(ns) \right] - \int_{\mathsf{p}}^{1} f(x) \Theta_A(x-\mathsf{p}) d \zeta^n_s(x)
\right| \Rightarrow 0.
\end{equation}

First, we prove (\ref{eq:generatorconvg3a}). We provide a sketch. Define a function $G$ by setting
\begin{equation}
 G(x) =
\begin{cases} f(x) \Theta_A(x-\mathsf{p}) & \text{if $\mathsf{p} < x \le 1 $,}
\\
0 &\text{if $0 \le x\le \mathsf{p}$.}
\end{cases}
\end{equation}
 If $G$ is continuous on
$[0,1]$, (\ref{eq:generatorconvg3a}) readily follows from \eqref{eq:convgzetaf} and the equivalence of Skorohod $J_1$ topology and the uniform topology on the space of continuous functions. If $G$ is
not continuous, one can similarly prove (\ref{eq:generatorconvg3a}) by constructing a family of
continuous functions $\{G_{\epsilon}: \epsilon>0 \}$ such that for each small $\epsilon>0$, $G$ and $G_{\epsilon}$ are equal except on
a small interval $(\mathsf{p}- \epsilon, \mathsf{p}]$. We omit further details.

Next, we prove (\ref{eq:generatorconvg3c}). We also provide a sketch.
Using Assumption~\ref{ass:rateconverg1} and the definition of $\zeta^n$ in (\ref{eq:empirialshape}), it is equivalent to show \text{as $n \rightarrow \infty$}
\begin{multline}\label{eq:generatorconvg2'}
\sup_{0 \le s \le T}\Bigg|
\sum_{\frac{k}{n}>\frac{p_B^n(ns)}{n}}\left[\frac{1}{n^{2}}f\left(\frac{k}{n}\right)\Theta_A \left(\frac{k}{n}
-\frac{p_B^n(ns)}{n}\right) \cdot \mathcal{X}^n_k(ns) \right] \\
- \sum_{\frac{k}{n}>
\mathsf{p}}\left[\frac{1}{n^{2}}f\left(\frac{k}{n}\right)\Theta_A\left(\frac{k}{n}
-\mathsf{p}\right) \cdot \mathcal{X}^n_k(ns) \right] \Bigg|
\Rightarrow 0.
\end{multline}
This can be readily verifed from \eqref{eq:askpriceconvg}, \eqref{eq:bidpriceconvg} and Lemma~\ref{lem:totalorder}.  We omit further details.

Similarly we can show {as $n \rightarrow \infty$},
\begin{equation} \label{eq:generatorconvg4}
\sup_{0 \le s \le T}
\left|\sum_{k<p_A^n(ns)}\left[\frac{1}{n}f\left(\frac{k}{n}\right)\Theta_B^n({p_A^n(ns)-x})
\cdot \mathcal{X}^n_k(ns) \right] - \int_{0}^{\mathsf{p}} f(x) \Theta_B(\mathsf{p}-x) d \zeta_s(x)
\right| \Rightarrow 0.
\end{equation}

Finally, using Assumption~\ref{ass:rateconverg1} and the fact that
$f$ is bounded on $[0,1]$, we obtain that there exists some constant $C$ such that
\begin{equation} \label{eq:generatorconvg5}
\sup_{0 \le s \le T}n \left|\frac{1
}{n^{2}}\left(f(\frac{p^n_B(u)}{n}) \Upsilon_A^n -f(\frac{p^n_A(u)}{n}) \Upsilon_B^n
\right) \right|
\le \frac{C}{n^{1- \kappa}} \Rightarrow 0 \quad \text{as $n
\rightarrow \infty$.}
\end{equation}

On combining (\ref{eq:generatorconvg1}), (\ref{eq:generatorconvg2}),
(\ref{eq:generatorconvg3}), (\ref{eq:generatorconvg4}), and
(\ref{eq:generatorconvg5}), we obtain (\ref{eq:generatorconvg}).
\end{proof}

\clearpage

\bibliographystyle{plain}

\end{document}